\documentclass{article}

\usepackage{PRIMEarxiv}

\usepackage[utf8]{inputenc} 
\usepackage[T1]{fontenc}    
\usepackage{hyperref}       
\usepackage{url}            
\usepackage{booktabs}       
\usepackage{amsfonts}       
\usepackage{nicefrac}       
\usepackage{microtype}      
\usepackage{lipsum}
\usepackage{fancyhdr}       
\usepackage{graphicx}       
\usepackage{amssymb,amsthm,amsmath}
\usepackage{xcolor,paralist,hyperref,titlesec,fancyhdr,etoolbox}
\usepackage{xspace}
\usepackage[ruled]{algorithm}
\usepackage[noend]{algorithmic}
\usepackage{caption}
\newtheorem{theorem}{Theorem}
\newtheorem{definition}{Definition}

\newtheorem{lemma}{Lemma}
\newtheorem{claim}{Claim}

\newtheorem{corollary}{Corollary}

\newcommand{\cB}{{\mathcal B}\xspace}
\newcommand{\cG}{{\mathcal G}\xspace}
\title{Approximation Algorithms for Clustering with Minimum Sum of Radii,  Diameters, and Squared Radii
}

\author{
  Z Friggstad, M Jamshidian \\
  Computing Science Department \\
  University of Alberta \\
  Edmonton, Alberta, Canada\\
  \texttt{\{zacharyf, mjamshidian\}@ualberta.ca} \\
}

\begin{document}
\maketitle

\begin{abstract}
In this study, we present an improved approximation algorithm for three related problems.
   In the Minimum Sum of Radii clustering problem (MSR), we aim to select $k$ balls in a metric space to cover all points while minimizing the sum of the radii. 
   In the Minimum Sum of Diameters clustering problem (MSD), we are to pick $k$ clusters to cover all the points such that sum of diameters of all the clusters is minimized. 
   At last, in the Minimum Sum of Squared Radii problem (MSSR), the goal is to choose $k$ balls, similar to MSR.
   However in MSSR, the goal is to minimize the sum of squares of radii of the balls. 
   We present a 3.389-approximation for MSR and a 6.546-approximation for MSD, improving over respective 3.504 and 7.008 developed by Charikar and Panigrahy (2001). 
   In particular, our guarantee for MSD is better than twice our guarantee for MSR.
   In the case of MSSR, the best known approximation guarantee is $4\cdot(540)^{2}$ based on the work of  Bhowmick, Inamdar, and Varadarajan in their general analysis of the $t$-Metric Multicover Problem. 
   At last with our analysis, we get a 11.078-approximation algorithm for Minimum Sum of Squared Radii.
\end{abstract}


\section{Introduction}
Clustering, as one of the fundamental problems in information technology, has been studied in computing science and several other fields to a great extent.
Different methods of clustering have been used significantly in data mining, bio-informatics, pattern recognition, computer vision, etc. 
The goal of clustering is to partition a set of data points into partitions, called clusters.
Many of clustering problems involve finding $k$ cluster centers and a mapping $\sigma$ from data points to the centers to minimize some objective function.
One of the most studied such objective functions is \textsc{$k$-Center} which minimizes the maximum diameter (or radius) \cite{DYER1985,dorit1985}. 
Another example is the \textsc{$k$-Median} problem which aims to minimize sum of distances from data points to their centers, as extensively studied in \cite{charikar1999,charikar2002,jain2001,kariv1979,lin1992}.

Another model of clustering is the \textsc{Facility Location Problem}, in which a set of clients are connected to facilities such that the cost of opening all facility and sum of distances from clients to facilities are minimized. 
Studies on FCL have led to significant advances in understanding the problem and its objective function, as discussed in \cite{charikar1999,chudak1999_c,chudak2003,Guha1999,jain2001,madhukar2000}.

In the late 20th century, \cite{Hansen1987} introduced a phenomenon called \textit{Dissection Effect} in the $k$-center problem where nodes that belong to the same cluster are placed in different ones as otherwise it would have increased the objective function. 
This motivates introduction of the ``Minimum Sum Diameters" (MSD) problem where the objective function is to minimize sum of diameters of all clusters. 
\cite{monma91} showed that the modified objective function is useful as it reduces the dissection effect. 

In this paper, we focus on a different objective function for clustering that is more center-focused in that the cost of a cluster is the radius of the ball used to cover that cluster. Specifically, we study the following problem.
\begin{definition}
In the \textsc{Minimum Sum of Radii} problem (MSR), we are given a set $X$ of $n$ points in a metric space with distances $d$ and a positive integer $k$. We are to select centers $C \subseteq X$, $|C| \leq k$ and assign each $i \in C$ a {\em radius} $r_i$
so that each $j \in X$ lies within distance $r_i$ of at least one $i \in C$ (i.e. $d(j,i) \leq r_i$). The goal is to minimize the total radii, i.e. $\sum_{i \in C} r_i$.
\end{definition}

That is, we want to cover $X$ using at most $k$ balls with minimum total radius. For example, perhaps we want to broadcast messages to all points by selecting $k$ sources with minimum total broadcast radius.

We also consider the related problem to minimize sum of diameters of the clusters chosen. Note this variant is simply about partitioning the point set, there are no centers involved.
\begin{definition}
In the \textsc{Minimums Sum of Diameters} problem (MSD), the input is the same as in MSR and our goal is to partition the points into $k$ clusters $X_1, X_2, \ldots, X_k$ to minimize $\sum_{i = 1}^k \max_{j,j' \in X_i} d(j,j')$, the sum of the diameters of the clusters.
\end{definition}

It is easy to see that an $\alpha$-approximation algorithm for MSR yields a $2\alpha$-approximation algorithm for MSD. That is, if $OPT_R$ denotes the optimum MSR solution cost and $OPT_D$ an optimum MSD solution cost, we have $OPT_R \leq OPT_D$ because in the optimum MSD solution we could pick any point from each cluster to act as its center (with radius equal to the diameter of the cluster). So if we have an MSR solution with cost at most $\alpha \cdot OPT_R$, then if we define clusters $X_i$ by sending each point to some center whose ball covers that point, the diameter of cluster $i$ would be $\leq 2 \cdot r_i$ so the sum of diameters would then be at most $2\alpha \cdot OPT_R \leq 2\alpha \cdot OPT_D$.

On the other hand, \cite{Doddi2000} showed that there is no $(2-\epsilon)$-approximation algorithm for MSD unless $P=NP$.
Their algorithm obtains a logarithmic approximation with a constant factor increase in number of clusters.
For a constant number of clusters, they give a $2$ approximation.

As an intuition to the MSR problem, centers can be thought of as prospective mobile tower locations, whereas the points in $\mathcal{X}$ can be thought of as client locations. A tower can be set up in such a way that it can service consumers within a given radius. However, the cost of service rises with the broadcast distance travelled. 
The goal is to serve all clients at the lowest possible cost. 

When calculating the amount of energy required for wireless transmission, it is typical to think about the cost function to be Minimum Sum of Squared Radii.
In reality, it requires power proportionate to $r^2$ to broadcast up to a certain radius $r$. This inspires a form of MSR in which we want to reduce the sum of the squares of the radii, i.e. the total broadcast power. The Minimum Sum of Squared Radii (MSSR) issue is what we refer to as.

A multi-cover variant, defined in \cite{bhowmick2017metric}, has also been considered.
\begin{definition}
In the \textsc{$t$ Metric Multi Cover} problem ($t$-MMC), the input is the sets $X$ and $Y$, representing Data Points and centers candidate set and our goal is to pick $t$ pairs of centers and radii $(c_{1},r_{1}),\cdots,(c_{t}, r_{t})$ such that each data point is present in at least $k$ of the balls while minimizing the objective function $\sum_{t}(r_{i})^{\alpha}$ where $\alpha$ and $k$ are given in the input.
\end{definition}
Previous works on the more general problem of $t$-MMC have established the guarantee of $4\cdot (540)^2$ \cite{Bhowmick2021}. We note that this was the best approximation for the natural special case MSSR before our work.

At last, the related problems in Euclidean space have received attention over the years. 
There is an exact polynomial-time algorithm for the Euclidean MSR problem as discussed in \cite{Gibson2009}.
Here we aim to consider the modification of the problem in a high-dimensional Euclidean space with an added flexibility to the center picking procedure. 
That is, one can pick any point in the space as the center. 
In what follows, we describe at details the improvements established in the approximation factors of the aforementioned problems, as well as introducing the novel and generic method for the bi-point analysis of the problems.

\subsection{Our Contribution}

In this paper, we first present an improved polynomial-time approximation algorithm for MSR. Specifically, we prove the following.
\begin{theorem}\label{thm:msr_main}
There is a polynomial-time 3.389-approximation for MSR.
\end{theorem}
We obtain this primarily by refining a so-called bi-point rounding step from \cite{charika04}.
That is, our improvement for MSR mainly focuses in the last phase of the algorithm in  \cite{charika04} which combines two subsets of balls that, together, open an {\em average} of $k$ centers and whose average cost is low. Their algorithm focuses on selecting $k$ of the centers from these two subsets. We expand the set of possible centers to choose and consider some that may not be centers in the averaging of the two subsets.

 We also present an alternative method for obtaining these two subsets of balls by considering a simple rounding of a linear programming (LP) relaxation, the Lagrangian relaxation of the problem obtained by relaxing the constraint that at most $k$ centers are chosen, rather than the primal-dual technique used in \cite{charika04}.The rounding algorithm is incredibly simple and we employ fairly generic arguments to convert it to a bi-point solution for a single Lagrangian multiplier $\lambda$, this approach may be of independent interest as it should be easy to adapt to other settings where one wants to get a bi-point solution where both points are obtained from a common Lagrangian value $\lambda$, as long as the LMP approximation is from direct LP rounding. However, we emphasize that we could work directly with their primal-dual approach.

Our second result is an improved MSD approximation that does not just use our MSR approximation as a black box.
\begin{theorem}
There is a polynomial-time 6.546-approximation for MSD.
\end{theorem}
In particular, notice the guarantee is better than twice our approximation guarantee for MSR. This is obtained through a variation of our new ideas behind our MSR approximation.

Finally, we get to discuss the MSSR problem, using the same machinery. 
The algorithm makes use of the same bi-point rounding and placing of the centers. 
At last, we have the following result. 
\begin{theorem}
There is a polynomial-time 11.078-approximation for MSSR.
\end{theorem} 
Note that this is best known result for the MSSR case. 
As mentioned, objective functions with powers greater than 1 has been studied in other works. 
Although, their corresponding problems are more strict and hence, they establish a larger approximation factor. 
One can also analyze Charikar's approach to get a constant approximation.

\subsection{Other Related Work}\label{sec:priorwork}
Gibson et al. show MSR is {\bf NP}-hard even in metrics with constant doubling dimension or shortest-path metrics of edge-weighted planar graphs \cite{Gibson2009}. 
In polynomial time, the best approximation algorithm is the stated 3.504 approximation by Charikar and Panigrahy. \cite{charika04}. 
Interestingly, \cite{Gibson2009} show that MSR can be solved exactly in $n^{O(\log n \cdot \log \Gamma)}$ where $\Gamma$ is the {\em aspect ratio} of the metric (maximum distance divided by minimum nonzero distance). 
Using Randomized Algorithm, this yields a quasi-PTAS for MSR: i.e. a $(1+\epsilon)$-approximation with running time $n^{O(\log 1/\epsilon + \log^2 n)}$. 
The main idea that underlies this result is that if we probabilistically partition the metric into sets with at most half the original diameter, then with high probability only $O(\log{n})$ balls in the optimal MSR solution are ``cut'' by the partition.

We note that an improved $(3+\epsilon)$-approximation for MSR approximation has designed by Buchem et. al. \cite{buchem24} since the preliminary version of our results first appeared \cite{esa22}. Like our work, they consider an LP relaxation for MSR. Our algorithm, like the previous one by Charikar and Panigrahy \cite{charika04}, considers the Lagrangian relaxation of the constraint asserting at most $k$ centers are open and rounds a bi-point solution (two integer solutions to the Lagrangian relaxation that can be averaged to find a near-optimal LP solution fractionally opening $k$ centers) to this Lagrangian relaxation, though with more scrutiny than in \cite{charika04}.
On the other hand, \cite{buchem24} works directly with the original LP relaxation and describes a novel primal-dual approach which obtains two related clusterings, which also be viewed as a bi-point solution. Their algorithm finishes by leveraging the structure of their bi-point solution to obtain a feasible MSR solution with the improved approximation guarantee. We emphasize that their approach yields improvements for MSR and variants with outliers and lower bounds but they do not describe any improved approximations for MSD nor MSSR.

Given that a quasi-PTAS exists for MSR, a major open problem is to design a true PTAS for MSR, or perhaps to demonstrate that a PTAS for MSR is not possible under something like the Strong Exponential Time Hypothesis. For now, it is of interest to get improved constant-factor approximations for MSR. By way of analogy, the {\em unsplittable flow} problem was known to admit a quasi-PTAS \cite{Bansal2006,UFP}
However, improved constant-factor approximations were subsequently produced \cite{UFP2,UFP3,UFP1}, that is until a PTAS was finally found by Grandoni et al. \cite{UFPPTAS}.

 Doddi et al. show that unless {\bf P} = {\bf NP}, there is no $(2-\epsilon)$-approximation for MSD for any $\epsilon > 0$ even if the metric is the shortest path metric of an unweighted graph \cite{Doddi2000}. Prior to our work, the best approximation for MSD is simply twice the best polynomial-time approximation for MSR, i.e. $2 \cdot 3.504 = 7.008$ using the approximation for MSR from \cite{charika04}.

MSR and MSD have been studied in special cases as well. In constant-dimensional Euclidean metrics, MSR can be solved exactly in polynomial time \cite{GibsonEuclidean}. This is particularly interesting in light of the fact that MSR is hard in doubling metrics. For MSD in constant-dimensional Euclidean metrics, if $k$ is also regarded as a constant then MSD can be solved exactly \cite{Capoyleas90geometricclusterings}. In general metrics with $k = 2$, MSD can be solved exactly by observing that if we are given the diameters of the two clusters, we can use 2SAT to determine if we can place the points in these clusters while respecting the diameters \cite{Hansen1987}. However, MSD is {\bf NP}-hard for even $k = 3$ as it captures the problem of determining if an unweighted graph can be partitioned into 3 cliques. Finally, if one does not allow balls with radius 0 in the solution, MSR can be solved in polynomial time in shortest path metrics of unweighted graphs \cite{Behsaz,Lokshtanov}.

In the variant of MSR with outliers, we are permitted to discard up to $m$ clients before optimally clustering the rest. A 12.365-approximation was first presented by Ahmadian and Swamy \cite{ahmadian2016} and the new $(3+\epsilon)$-approximation for MSR in \cite{buchem24}
applies to the setting with outliers.

The lower-bounded setting has also been studied. That is, for each possible center $i$ there is a bound $L_i$ on the number of data points that must be assigned to that center.
A 3.82-approximation was given in \cite{ahmadian2016}, where they also obtain a 12.365-approximation for the general case with both outliers and center lower bounds.
We note this general case now has an improved $(3.5 + \epsilon)$-approximation \cite{buchem24}.


Bhowmick, Inamdar, and  Varadarajan \cite{Bhowmick2021} study a fault tolerant version.
They discuss the Metric Multi-Cover problem, in which we have a set of data points $\mathcal{X}$ and set of  candidate facilities $\mathcal{Y}$ and a demand value $k$ for all the data points.
The goal here is to pick pairs of $(i,r)$, as many as we wish but at most one per center $i$, to minimize the sum of $r^\alpha$ (for some given $\alpha$) for chosen balls while each data point has to be covered by at least $k$ balls.
This modification of the problem admits an approximation factor of $2 \cdot (108)^\alpha$.
If we add the constraint to use at most $t$ centers (called $t$-MMC), their work gives an approximation factor of $4\cdot(540)^\alpha$.
On another note, if each data point has specific demand, the factor changes to $2\cdot (144)^\alpha$.
Additionally, a QPTAS also exists for the standard Metric Multi-Cover problem based in \cite{bandyapadhyay2016}.

For $t$-MMC in Euclidean spaces, a $(23.02 + 63.95(t-1))$-approximation is known  \cite{Abu2011}. 
They also consider a non-discrete version of MMC, where there are areas to cover rather than discrete data points.
Their work proves a $63.94 + 177.64(t-1)$ approximation factor. 
For a client specific version of MMC in the Euclidean, \cite{shike2013} provides a $4(27\sqrt{2})^\alpha$-approximation. 
Also, \cite{Liu2021} studies $t$-MMC problem with a penalty added for each uncovered data point. 
Then the objective function is to minimize minimum sum of radii plus the total penalty of uncovered points. 
They give a $3^\alpha + r_{max}^{\alpha}$ approximation algorithm in Euclidean space. 

\section{Minimum Sum of Radii}
We first focus on the MSR problem and later extend the procedure to MSD and MSSR.
Recall $X$ is the set of all data points in a metric with distances given by $d$ and we are to pick a most $k$-centers and assign a radius to each such that each point in $X$ lies within the assigned radius of one of the chosen centers.

From now on we let $n := |X|$.
We may assume $d(i,i') > 0$ for distinct $i,i' \in X$ by removing all but one point from each group of colocated points.
A {\bf ball} in $X$ is a set of the form $B(i,r) = \{j \in X : d(i,j) \leq r\}$ for some point $i \in X$ and radius $r \geq 0$. 
We refer to a pair $(i,r)$ as a ball with the understanding it is referring to the set $B(i,r)$. 
We view a solution as a collection $\mathcal B$ of pairs $(i,r), i \in X, r \geq 0$ describing the centers and radii of the balls. 
For such a subset, we let $cost(\mathcal B) = \sum_{(i,r) \in \mathcal B} r$ be the total radii of these balls.

Fix some small constant $\epsilon > 0$ such that $1/\epsilon$ is an integer.
Note that a smaller $\epsilon$ leads to a better guarantee with increased (but still polynomial) running time.
We will be able to pick a small enough $\epsilon$ such that it will hide in the approximation guarantee. This is why it is not mentioned in the statement of Theorem \ref{thm:msr_main}.
We assume $k > 1/\epsilon$, otherwise we can simply use brute force to find the optimum solution in $n^{O(1/\epsilon)}$ time.

Here we provide a detailed explanation for the MSR algorithm and all its subroutines. 
The main contribution happens at the last phase of algorithm, where we establish a better approximation factor.
Our algorithm for MSR is summarized in Algorithm \ref{alg:msr} at the end of this section, though it makes reference to a fundamental subroutine to find our ``bi-point'' solution that we describe later. 
By bi-point, we simply mean two subsets of balls $\cB_1, \cB_2$ with $|\cB_1| \geq k \geq |\cB_2|$ so some averaging of these balls looks like a feasible fractional solution using exactly $k$ balls.


\subsection{Step 1: Guessing the Largest Balls}\label{sec:guess}

Let $\mathcal B^*$ denote some fixed optimum solution with $OPT := cost(\mathcal B^*)$. 
Among all optimal solutions, we assume $\mathcal B^*$ has the fewest balls. Thus, for distinct $(i,r), (i',r') \in \cB^*$ we have that $i' \notin B(i,r)$ since, otherwise, $\mathcal \cB^* = \{(i,r), (i',r')\}) \cup \{(i,r+r')\}$ is another optimal solution with even fewer balls.

Similar to \cite{charika04}, we guess the $1/\epsilon$ largest balls in $\mathcal{B^*}$ by trying each subset $\cB'$ of $1/\epsilon$ balls and proceeding with the algorithm we describe in the rest of this section.
That is, let $\mathcal{B'} \subseteq \mathcal B^{*}$ be such that $|\mathcal{ B'}| = 1/\epsilon$ and $r \leq r'$ for each $(i,r) \in \mathcal{B^{*}} - \mathcal{B'}$ and $(i',r') \in \mathcal{B'}$.

Let $R_m$ be the minimum radius of a ball in $\mathcal B'$. Remember that since we already have guessed the largest balls and the sum of their radii is at most $OPT$, then $R_m \leq \epsilon \cdot OPT$. 
Let $k' := k - 1/\epsilon$ and note $k'$ is an upper bound on the number of balls in $\cB^* - \cB'$.

We now restrict ourselves to the instance with points $X' := X - \cup_{(i,r) \in \mathcal B'} B(i,r)$ not yet covered by $\mathcal B'$.
Since no center of a ball in $\cB^*$ is contained within another ball from $\cB^*$, the remaining balls in $\cB^* - \cB'$ are also centered in $X'$. 
We will let $OPT' = OPT - \sum_{(i,r) \in \mathcal B'} r$ denote the optimal solution value to this restricted instance. 
The solution $\mathcal B^* - \mathcal B'$ for this instance satisfies $r \leq R_m \leq \epsilon \cdot OPT$ for any $(i,r) \in \mathcal B^* - \mathcal B'$. We also assume $|X'| > k'$, otherwise we just open zero-radius balls at each point in $\mathcal X'$.

Our guessing step must perform a ``precheck'' for this guess as follows before proceeding.
We use a $2$-approximation for the classic \textsc{$k'$-center} problem \cite{dorit1985} on the metric restricted to $X'$.
If the solution returned has radius $> 2 \cdot R_m$, then we reject this guess $\cB'$. This is valid because we know for a correct guess that the remaining points can each be covered using at most $k'$ balls each with radius at most $R_m$.
From now on, we let $\mathcal A$ denote the $k'$ centers returned by this approximation: so each $j \in X'$ lies in at least one ball of the form $B(i, 2 \cdot R_m)$ for some $i \in \mathcal A$.


The preceding discussion is summarized below.
\begin{theorem}\label{thm:summary}
There is an optimal solution set $\mathcal{B}^*$ with value $OPT$ such that the following holds. Let $\mathcal{B}^{'}$ be the set of $\frac{1}{\epsilon}$ largest balls in $\mathcal{B}^{*}$.
and $R_{m}$ is the minimum radius in $\mathcal{B}'$.

Then $R_{m} \leq \epsilon \cdot OPT$.
Furthermore, let $X' = X - \cup_{(i,r) \in \mathcal{B}^{'}} B(i,r)$, the set of data points not covered by $\mathcal{B'}$.
Finally, if $\mathcal{A}$ is a set of centers obtained by running a $2$-approximation for the $(k-1/\epsilon)$-center problem in the metric restricted to $X'$, then $d(j,\mathcal A) \leq 2 \cdot R_m$ for each $j \in X'$.

\end{theorem}

When analyzing the rest of the algorithm, we will assume $\cB'$ is guessed correctly, i.e. $\cB' \subseteq \cB^*$ and all $(i,r) \in \cB^* - \cB'$ have $r \leq R_m$. 
Our final solution will be the minimum-cost solution found over all guesses $\cB'$ that were not rejected, so the final solution's cost will be at most the cost of the solution found when $\cB'$ was guessed correctly.


\subsection{Step 2: Getting a Bi-point Solution}
The output from this step is similar to \cite{charika04}.
We remark that their approach would suffice for our purposes, except there would be yet another ``$\epsilon$'' introduced in the approximation guarantee with their technique.
We will be following a different approach primarily to show there is a simple and direct LP rounding routine and, more importantly, to give a generic procedure that is likely to apply to most LP-rounding LMP approximations to get a bi-point solution where both points can be compared with the optimal LP solution for a single value $\lambda$ for the Lagrangian variable.
 
First, consider the following LP relaxation for MSR where we have a variable $x_{(i,r)}$ for each of the $O(n^2)$ possible balls, that is for each $i \in X'$ and $0 \leq r \leq 2 \cdot R_m$ such that $r = d(i,j)$ for some $j \in X'$ (since we only need to consider balls of radius equal to the maximum distance of a client covered by the ball). The factor of $2$ in the maximum radius $2 \cdot R_m$ is because we want to allow balls in the $k$-center solution mentioned in Theorem \ref{thm:summary}, this will be important when we round the Lagrangian relaxation of the LP below.
\[
\begin{array}{rrcll}
{\bf min} &\sum_{(i,r)} r \cdot x_{(i,r)} & \tag{\bf LP} \label{lp:primal-bf}\\
{\bf s.t.} & \sum_{(i,r) : j \in B(i,r)} x_{(i,r)} & \geq & 1 & \forall~ j \in X' \\
& \sum x_{(i,r)}  & \leq & k \\
& x & \geq & 0
\end{array}
\]

For a value $\lambda \geq 0$, \ref{lp:primal} is the linear program that results by considering the Lagrangian relaxation of MSR through ``Lagrangifying'' the cardinality constraint on the number of open centers. 
That is, the LP has variables for each possible ball we may add except instead of restricting the number of balls to be at most $k'$, we simply pay $\lambda$ for each ball.
\[
\begin{array}{rrcll}
{\bf min} &\sum_{(i,r)} (r+\lambda) \cdot x_{(i,r)} & \tag{\bf LP($\lambda$)} \label{lp:primal}\\
{\bf s.t.} & \sum_{(i,r) : j \in B(i,r)} x_{(i,r)} & \geq & 1 & \forall~ j \in X' \\
& x & \geq & 0
\end{array}
\]

We also consider the dual of \ref{lp:primal}.
\[
\begin{array}{rrcll}
{\bf max} & \sum_{j \in X'} y_j \tag{\bf DUAL($\lambda$)} \label{lp:dual} \\
{\bf s.t.} & \sum_{j \in B(i,r) \cap X'} y_j & \leq & r + \lambda & \forall~ (i,r), r \leq R_m \\
& y & \geq & 0
\end{array}
\]

The following is standard.
\begin{lemma}\label{lem:lmp}
For any $\lambda \geq 0$, let $OPT_{\text{\ref{lp:primal}}}$ denote the optimum value of \ref{lp:primal}. Then $OPT_{LP(\lambda)} - \lambda \cdot k' \leq OPT'$.
\end{lemma}
\begin{proof}
Using the natural $\{0,1\}$-solution corresponding to the optimal MSR solution (which uses at most $k'$ balls), we see there is a feasible solution for \eqref{lp:primal} with cost at most $OPT' + k' \cdot \lambda$.
That is, $OPT_{LP(\lambda)} - \lambda \cdot k' \leq OPT'$.
%
\end{proof}

\begin{theorem}\label{thm:round}
Let $\lambda \geq 0$ and let $x'$ be an optional solution to \eqref{lp:primal}. There is a polynomial-time algorithm ROUND($x'$) that returns a set of balls $\cB$ such that $x'_{i,r} > 0$ for each $B(i,r) \in \cB$.
Furthermore, the balls in $\cB$ are pairwise-disjoint and if we triple the radius of each ball in $\cB$ then all data points would be covered.
\end{theorem}


The proof of Theorem \ref{thm:round} will be presented in Section \ref{sec:lmp}. 
For now, we note how the solution returned is by ROUND$(x'$) is cheap when applied to an optimal LP solution $x'$.


\begin{corollary}\label{cor:round}
Let $\lambda \geq 0$ and let $\overline{\mathcal B} = \text{ROUND}(x')$ where $x'$ is an optimal solution to \ref{lp:primal}. Then $cost(\overline{\mathcal B}) + \lambda \cdot |\overline{\mathcal B}| \leq OPT_{\text{\ref{lp:primal}}}$.
\end{corollary}
\begin{proof}
Let $y'$ be an optimal dual solution to \eqref{lp:primal}.
By complementary slackness, for each $B(i,r) \in \overline{\mathcal B}$ we have $r + \lambda = \sum_{j \in B(i,r)} y'_j$. Since the balls are disjoint, then $\sum_{B(i,r) \in \overline{\mathcal B}} \sum_{j \in B(i,r)} y'_j \leq \sum_{j' \in X'} y'_j = OPT_{\text{\ref{lp:primal}}}$.
%
\end{proof}


The next theorem, also proven in Section \ref{sec:lmp}, summarizes our LMP approximation result that is obtained by using a ``binary search'' over $\lambda$ and using the algorithm from Theorem \ref{thm:round} to determine if $\lambda$ was too small or too big.
The following theorem is required for the analysis later studied.
\begin{theorem}\label{thm:bipoint}
There is a polynomial-time algorithm that will compute a single value $\lambda \geq 0$ and two sets of balls $\cB_1, \cB_2$ having respective sizes $k_1, k_2$ where $k_1 \geq k' \geq k_2$.
Furthermore, for every $(i,r) \in \cB_1$, there is some $(i',r') \in \cB_2$ such that $B(i,r) \cap B(i',r') \neq \emptyset$. Finally, for both $\ell = 1$ and $\ell = 2$ we have the following properties:
\begin{itemize}
\item for each $(i,r) \in \cB_\ell$, we have $r \leq 2 \cdot R_m$ ($R_{m}$ is the smallest radius in the set of guessed balls),
\item tripling the radii of each $(i,r) \in \cB_\ell$ will cover $X'$, i.e. for each $j \in X'$ there is some $(i,r) \in \cB_\ell$ such that $j \in B(i,3 \cdot r)$, and
\item $cost(\cB_\ell) + \lambda \cdot k_\ell \leq OPT_{\text{\ref{lp:primal}}}$
\end{itemize}
\end{theorem}


As a warm-up, we consider the case when one of $k_{1}$ or $k_{2}$ is equal to $k^{'}$.
\begin{lemma}\label{lem:simple}
If for $\ell = 1$ or $\ell = 2$ we have $|\cB_\ell| = k'$, then tripling the radii of the balls in $\cB_\ell$ yields a feasible solution with cost at most $3 \cdot OPT'$ for MSR.
\end{lemma}
\begin{proof}
Let $x_\ell$ be such that $\text{ROUND}(x)$ produced $\cB_\ell$.
Then
\[ cost(\cB_\ell) \leq OPT_{\text{\ref{lp:primal}}} -  \lambda \cdot |\cB_\ell| = OPT_{\text{\ref{lp:primal}}} -  \lambda \cdot k'   \leq OPT' \]
The first bound is by Corollary \ref{cor:round} and the second bound is by Lemma \ref{lem:lmp}.
By Theorem \ref{thm:round}, tripling the radii of balls in $\cB_\ell$ will cover all points. These tripled balls then have cost $3 \cdot cost(\cB_\ell) \leq 3 \cdot OPT'$.
\end{proof}

Our final algorithm in the next section works in the general case where $k_1 \geq k' \geq k_2$ and does not distinguish the cases $k_1 = k'$ and/or $k_2 = k'$. Lemma \ref{lem:simple} was simply a special case we considered to build intuition.


\subsection{Step 3: Combining Bi-point Solutions for MSR}\label{sec:combining}
Let $\lambda, \cB_1, \cB_2$ be the {\em bi-point solution} from Theorem \ref{thm:bipoint}. For brevity, let $C_1 = cost(\cB_1)$ and $C_2 = cost(\cB_2)$.
Since $k_1 \geq k' \geq k_2$, there are values $a,b \geq 0$ with $a+b = 1$ and $a \cdot k_1 + b \cdot k_2 = k'$. We fix these values throughout this section.

The following shows the {\em average} cost $C_1$ and $C_2$ is bounded by $OPT'$, the first inequality is by the last property listed in Theorem \ref{thm:bipoint} and the second by Lemma \ref{lem:lmp}.
\begin{equation}\label{eqn:bound}
a \cdot C_1 + b \cdot C_2 \leq a \cdot (OPT_{\text{\ref{lp:primal}}} - \lambda \cdot k_1) + b \cdot (OPT_{\text{\ref{lp:primal}}} - \lambda \cdot k_2) = OPT_{\text{\ref{lp:primal}}} - \lambda \cdot k' \leq OPT'
\end{equation}
The rest of our algorithm and analysis considers how to convert the two solutions $\cB_1, \cB_2$ to produce a feasible solution whose value is within a constant-factor of this averaging of $C_1, C_2$.
First, note tripling the radii in all balls in $\cB_2$ will produce a feasible solution as $k_2 \leq k'$, but it may be too expensive. So we will consider two different solutions and take the better of the two. The first is solution is what we just described: formally it is $\{(i,3r) : (i,r) \in \cB_2\}$, which is a feasible solution with cost $3 \cdot C_2$.

Constructing the second solution is our main deviation from the work in \cite{charika04}. Intuitively, we want to cover all points by using balls $(i,3\cdot r)$ for $(i,r) \in \cB_1$. The cheaper of this and the first solution can easily be show to have cost at most $3 \cdot OPT'$. The problem is that this could open more than $k'$ centers (if $k_1 > k'$). As in \cite{charika04}, we consolidate some of these balls into a single group based on their common intersection with some $(i', r') \in \cB_2$. We will select some groups and merge their balls into a single ball so the number of balls is at most $k'$. Our improved approximation is enabled by considering different ways to cover balls in a group using a single ball, \cite{charika04} only considered one possible way to cover a group with a single ball.

We now form groups. For each $(i,r) \in \cB_2$, we create a group $G_{(i,r)} \subseteq \cB_1$ as follows: for each $(i',r') \in \cB_1$, consider any single $(i,r) \in \cB_2$ such that $B(i,r) \cap B(i',r') \neq \emptyset$ and add $(i',r')$ to $G_{(i,r)}$. If multiple $(i,r) \in \cB_2$ satisfy this criteria, pick one arbitrarily. Let $\mathcal G = \{G_{(i,r)} : (i,r) \in \cB_2 \text{ s.t. } G_{(i,r)} \neq \emptyset\}$ be the collection of all nonempty groups formed this way, note $\mathcal G$ is a partitioning of $\cB_1$.

\subsubsection*{Covering a group with a single ball}
From here, the approach in \cite{charika04} would describe how to merge the balls in a group $G_{(i,r)} \in \cG$ simply by centering a new ball at $i$, and making its radius sufficiently large to cover all points covered by the tripled balls $B(i',3r')$ for $(i',r') \in G_{(i,r)}$. We consider choosing a different center when we consolidate the $\cB_1$ balls in a group. In fact, it suffices to simply pick the minimum-radius ball that covers the union of the tripled balls in a group. This ball can be centered at any point in $X'$. 
\begin{theorem}
If one decide to replace a group $G_{(i,r)}$ by a single ball, the cost of the ball is at most $\frac{11}{8}\cdot r + 3 \cdot cost(G_{(i,r)})$.
\end{theorem}
The exact choice of ball we use for the analysis depends on the composition of the group, namely the total and maximum radii of balls in $G_{(i,r)}$ versus the radius $r$ itself. In \cite{charika04}, the ball they select has cost at most $r + 4 \cdot cost(G_{(i,r)})$. While our analysis has a higher dependence on $r$, when considered as an alternative solution to the one that just triples all balls in $\cB_2$ we end up with a better overall solution.


\begin{figure}
    \centering
    \includegraphics[width=0.75\linewidth]{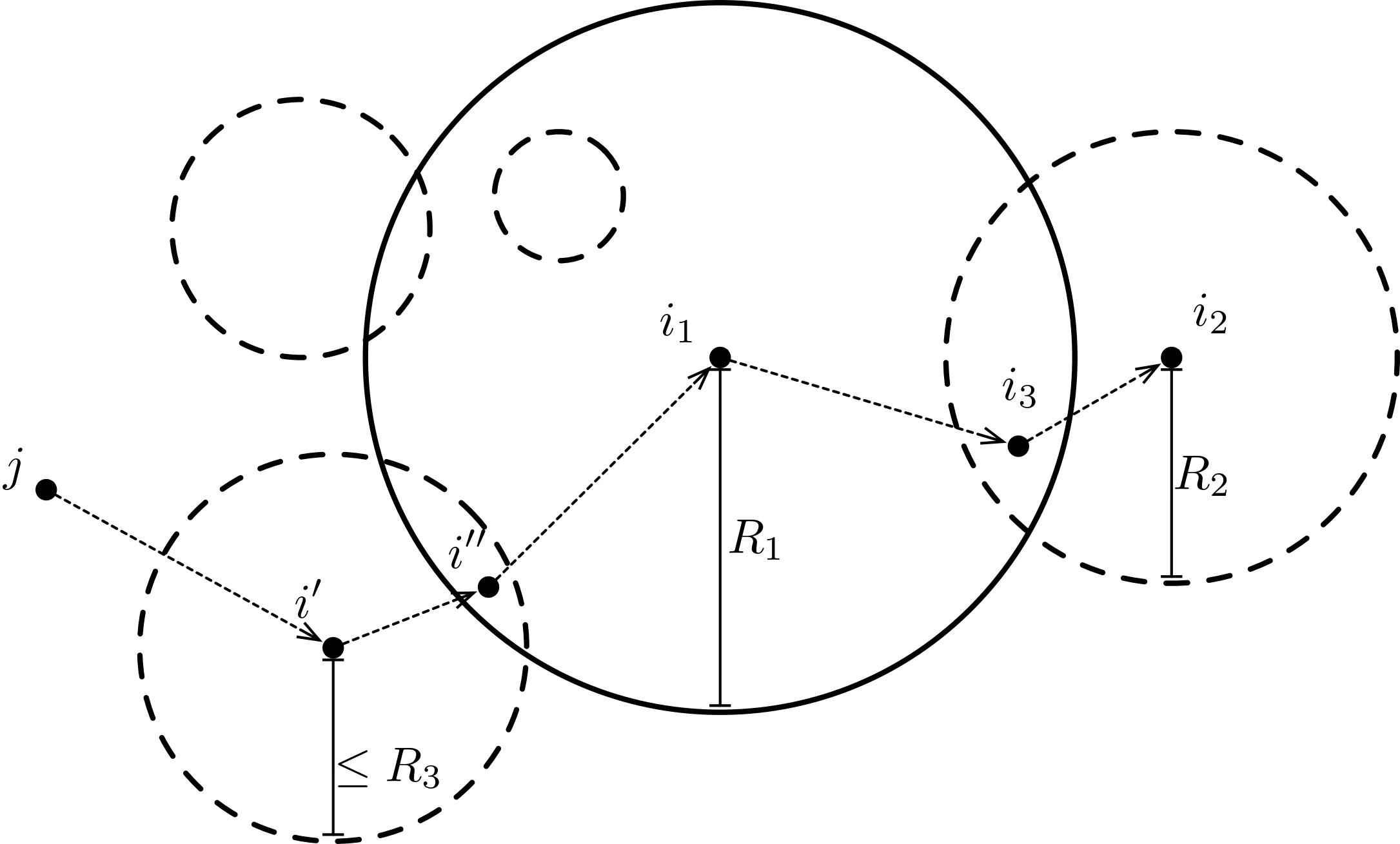}
    \caption{A depiction of a group $G_{(i_1,R_1)}$. The solid ball is $B(i_1,R_1)$ and the dashed balls are those in $G_{(i_1,R_1)}$. Point $j$ is covered by tripling the ball centered at $i'$. The dashed path depicts the way we bound $d(j,i_2)$ in the second part of the case {\bf Centering at $i_2$}.}
    \label{fig:pic2}
\end{figure}


For now, fix a single group $G_{(i,r)} \in \cG$.
Let $R_1$ denote $r$, $R_2$ be the maximum radius of a ball in $G_{(i,r)}$ and $R_3$ be the maximum radius among all other balls in $G_{(i,r)}$ apart from the one defining $R_2$. If $G_{(i,r)}$ has only one ball, then let $R_3 = 0$.
That is, $0 \leq R_3 \leq R_2$ but it could be that $R_3 = R_2$, i.e. there could be more than one ball from $G_{(i,r)}$ with maximum radius. We also let $i_1$ denote $i$, $i_2$ be the center of any particular ball with maximum radius in $G_{(i,r)}$, and $i_3$ be any single point in $B(i_1, R_1) \cap B(i_2, R_2)$. There is at least one since each ball in $G_{(i,r)}$ intersects $B(i,r)$ by construction of the groups.

Next we describe the radius of a ball that would be required if we centered it at one of $i_1, i_2$ or $i_3$. 
Consider any $j \in B(i', 3r')$ for some $(i',r') \in G_{(i,r)}$. Let $i''$ be any point in $B(i_1,r) \cap B(i',r')$.
We bound the distance of $j$ from each of $i_1, i_2$ and $i_3$ to see what radius would suffice for each of these three possible centers. Figure \ref{fig:pic2} depicts this group and one case of the analysis below.

\begin{itemize}
\item {\bf Centering at $i_1$}. Simply put,
\[ d(j,i_1) \leq d(j,i') + d(i',i'') + d(i'',i_1) \leq 3 \cdot R_2 + R_2 + R_1 = R_1 + 4\cdot R_2. \]
So radius $C^{(1)} := R_1 + 4\cdot R_2$ suffices if we choose $i_1$ as the center.
\item {\bf Centering at $i_2$}.
If $(i',r') = (i_2, R_2)$ then $d(j,i_2) \leq 3\cdot R_2$.
Otherwise, $r' \leq R_3$ and
\[ d(j,i_2) \leq d(j,i') + d(i',i'') + d(i'',i_1) + d(i_1,i_3) + d(i_3, i_2) \leq 3 \cdot R_3 + R_3 + R_1 + R_1 + R_2\] \[= 2\cdot R_1 + R_2 + 4 \cdot R_3. \]
So radius $C^{(2)} := \max\{3 \cdot R_2, 2\cdot R_1 + R_2 + 4 \cdot R_3\}$ suffices if we choose $i_2$ as the center.
\item {\bf Centering at $i_3$}.
If $(i',r') = (i_2, R_2)$ then $d(j,i_3) \leq d(j, i_2) + d(i_2, i_3) \leq 3\cdot R_2 + R_2 = 4\cdot R_2$.
Otherwise, $r' \leq R_3$ and we see
\[ d(j, i_3) \leq d(j, i') + d(i', i'') + d(i'', i_1) + d(i_1, i_3) \leq 3 \cdot R_3 + R_3 + R_1 + R_1 = 2\cdot R_1 + 4 \cdot R_3. \]
So radius $C^{(3)} := \max\{4 \cdot R_2, 2\cdot R_1 + 4 \cdot R_3\}$ suffices if we choose $i_3$ as the center.
\end{itemize}

With these bounds, we now describe how to choose a single ball covering the points covered by tripled balls in $G_{(i,r)}$ in a way that gives a good bound on the minimum-radius ball covering these points. The following cases employ particular constants to decide which center should be used, these have been optimized for our approach. 

The final bounds are stated to be of the form $3 \cdot C_{(i,r)}$ plus some multiple of $r$.
Let $C_{(i,r)} = \sum_{(i',r') \in G_{(i,r)}} r$ be the total radii of all balls in $G_{(i,r)}$. So $\sum_{G_{(i,r)} \in \mathcal G} C_{(i,r)} = cost(\cB_1) = C_1$.

\begin{itemize}
\item {\bf Case}: $R_3 > R_2/3$. Then the ball $B'_{(i,r)}$ is selected to be $B(i_1, C^{(1)})$. Note $4/3 \cdot R_2 < R_2 + R_3 \leq C_{(i,r)}$ so $C^{(1)} \leq r +3 \cdot C_{(i,r)}$.
\item {\bf Case}: $R_3 \leq R_2/3$ and $R_2 \geq \frac{6}{5} \cdot R_1$. The ball $B'_{(i,r)}$ is selected to be $B(i_2, C^{(2)})$.
Note $C^{(2)} \leq \frac{6}{5}  \cdot R_1 + 3 \cdot R_2 \leq \frac{6}{5}   \cdot r + 3 \cdot C_{(i,r)}$.
\item {\bf Case}: $R_3 \leq R_2/3$ and $\frac{6}{5} \cdot R_1 > R_2 \geq \frac{3}{8} \cdot R_1$. The ball $B'_{(i,r)}$ is selected to be $B(i_3, C^{(3)})$. Note
$C^{(3)} \leq \frac{11}{8} \cdot R_1 + 3 \cdot R_2 \leq \frac{11}{8}\cdot r + 3 \cdot C_{(i,r)}$.
\item {\bf Case}: $R_3 \leq R_2/3$ and $\frac{3}{8} \cdot R_1 > R_2$. The ball $B'_{(i,r)}$ is selected to be $B(i_1, C^{(1)})$. Note $C^{(1)} \leq \frac{11}{8} \cdot R_1 + 3 \cdot R_2 \leq \frac{11}{8}\cdot r + 3 \cdot C_{(i,r)}$.
\end{itemize}
In any case, we see that by selecting $B'_{(i,r)}$ optimally, the radius is at most $\frac{11}{8} \cdot r + 3 \cdot C_{(i,r)}$. Also, since $R_1, R_2, R_3 \leq 2 \cdot R_m$ by Theorem \ref{thm:bipoint}, then the radius of $B'_{(i,r)}$ is also seen to be at most, say, $14 \cdot R_m$.
That is because of two facts
\begin{enumerate}
    \item With additional preprocessing described in next chapter, a ball in the output of rounding has radius at most $2\cdot R_{m}$, and
    \item If we pick our center to be $i_{2}$, we might have the case where the radius is $2\cdot R_{1} + R_{2} + 4 \cdot R_{3}$.
\end{enumerate}
Then, by replacing $R_{1}$, $R_{2}$, and $R_{3}$ by $2 \cdot R_{m}$, the radius of the fractional group is bounded by $14 \cdot R_{m}$.

\subsubsection*{Choosing which groups to merge}
For each group $G_{(i,r)} \in \cG$, we consider two options. Either we select all balls in $G_{(i,r)}$ with triple their original radii (thus, with total cost $3 \cdot C_{(i,r)}$),
or we select the single ball $B'_{(i,r)}$.
We want to do this to minimize the resulting cost while ensuring the number of centers open is at most $k'$. To help with this, we consider the following linear program. For each $G_{(i,r)} \in \cG$ we have a variable $z_{(i,r)}$
where $z_{(i,r)} = 0$ corresponds to selecting the $|G_{(i,r)}|$ balls with triple their original radius and $z_{(i,r)} = 1$ corresponds to selecting the single ball $B'_{(i,r)}$. As noted in the previous section,
the radius of $B'_{(i,r)}$ is at most $\frac{11}{8} \cdot r + 3 \cdot C_{(i,r)}$ and also at most $14 \cdot R_m$.

\[
\begin{array}{rrcll}
{\bf minimize}: & \sum_{G_{(i,r)} \in \cG} (1-z_{(i,r)}) \cdot 3 \cdot C_{(i,r)} + z_{(i,r)} \cdot cost(\{B'_{(i,r)}\}) \tag{\bf{LP-Choose}} \label{lp:choose} \\
{\bf subject~to}: & \sum_{G_{(i,r)} \in \cG} \left((1-z_{(i,r)}) \cdot |G_{(i,r)}| + z_{(i,r)}\right) & \leq & k' \\
& z_{(i,r)} & \in & [0,1] & \forall~ G_{(i,r)} \in \cG
\end{array}
\]

Recall that $a,b$ are such that $a,b \geq 0, a + b = 1$ and $a \cdot k_1 + b \cdot k_2 = k'$. Thus, setting $z_{(i,r)} = a$ for each $G_{(i,r)} = 1$ is feasible as $1-z_{(i,r)} = b$, $\sum_{G_{(i,r)} \in \cG} |G_{(i,r)}| = k_2$, and there are at most $k'_1$ terms in this sum. The value of this LP solution is
\begin{equation*}
    \begin{split}
        \sum_{G_{(i,r)} \in \mathcal{G}} \left(3 \cdot b  + 3 \cdot a\right) C_{(i,r)} &= 3 \cdot \left(\sum_{G_{(i,r)} \in \mathcal{G}} C_{(i,r)}\right) + \frac{11}{8} \cdot b \cdot \left(\sum_{G_{(i,r)} \in \mathcal{G}} r\right) \\
        & = 3 \cdot C_{1} + \frac{11}{8} \cdot b \cdot C_{2}
    \end{split}
\end{equation*}
so the optimum solution to \ref{lp:choose} has value at most this much as well.

To consolidate the groups, compute an optimal extreme point to \ref{lp:choose}. Since all but one constraint are $[0,1]$-box constraints, there is at most one variable $z_{(i,r)}$ that does not take an integer value. Since $|G_{(i,r)}| \geq 1$, then setting $z_{(i,r)}$ to 1 yields a feasible solution whose cost increases by at most the radius of $B'_{(i,r)}$, which was observed to be at most $14 \cdot R_m \leq 14 \cdot \epsilon \cdot OPT$.

Summarizing,
\begin{lemma}
In polynomial time, we can compute a set of at most $k'$ balls with total radius at most $\frac{11}{8} \cdot b \cdot C_2 + 3 \cdot C_1 + 14 \cdot \epsilon \cdot OPT$ which cover all points in $X'$.
\end{lemma}

Finally, we can complete our analysis. Recall our simple solution of tripling the balls in $\cB_2$ has cost at most $3 \cdot C_2$ and the more involved solution jut described has cost at most $3 \cdot C_1 +  \frac{11}{8} \cdot a \cdot C_2 + 21 \cdot \epsilon \cdot OPT$. Now,
\[ \min\left\{3 \cdot C_2, 3 \cdot C_1 + \frac{11}{8} \cdot b \cdot C_2 \right\} \leq (1-d) \cdot 3 \cdot C_2 + d \cdot \left(b \cdot \frac{11}{8} \cdot C_2 + 3 \cdot C_1\right) \]
holds for any $0 \leq d \leq 1$. To maximize the latter, we set $d =  \frac{3(1-b)}{\frac{11}{8} \cdot b^{2} - \frac{11}{8} \cdot b + 3}$ and see the minimum of these two terms is at most
\[ \left(\frac{9}{\frac{11}{8} \cdot b^{2} - \frac{11}{8} \cdot b + 3}\right) \cdot (a C_{1} + b C_{2}) \leq  \left(\frac{9}{\frac{11}{8} \cdot b^{2} - \frac{11}{8} \cdot b + 3}\right) \cdot OPT' \]
where we have used bound \ref{eqn:bound} for the last step.

The worst case occurs at $b = \frac{1}{2}$, at which the bound becomes $\frac{85}{288} \cdot OPT'$. Thus, the cost of the solution is at most $\frac{288}{85} \cdot OPT' + 21 \cdot \epsilon \cdot OPT$.
Adding the balls $\mathcal B'$ we guessed to also cover the points in $X-X'$, we get get a solution covering all of $X$ with total radii at most
\[ cost(\mathcal B') + \frac{288}{85} \cdot OPT' + 14 \cdot \epsilon OPT = OPT - OPT' + \frac{288}{85} \cdot OPT' + 14 \cdot \epsilon \cdot OPT \leq 3.389 \cdot OPT \]
for sufficiently small $\epsilon$.

The entire algorithm for MSR that we have just presented is summarized in Algorithm \ref{alg:msr}.
\begin{algorithm}
\caption{MSR Approximation}
\begin{algorithmic} \label{alg:msr}
\STATE $\mathcal S \gets \emptyset$ ~~~~~~~\COMMENT{The set of all solutions seen over all guesses}

\FOR{each subset $\cB'_j$ of $1/\epsilon$ balls}
\STATE let $X',R_m$ be as described in Section \ref{sec:guess}
\STATE $(\mathcal A, R) \gets$ \textsc{$k$-center} 2-approximation on $X'$
\IF{$R > 2 \cdot R_m$}
\STATE {\bf reject} this guess $\cB'$ and continue with the next
\ENDIF
\STATE let $\cB_1, \cB_2, \lambda$ be the bi-point solution from Algorithm \ref{alg:bipoint} \COMMENT{see Theorem \ref{thm:bipoint}}
\STATE let $\mathcal G$ be the groups (a partitioning of $\cB_1$) described in Section \ref{sec:combining}
\STATE for each $G_{(i,r)} \in \cG$, let $B'_{(i,r)}$ be the cheapest ball covering $\cup_{(i',r') \in G_{(i,r)}} B(i',3 \cdot r')$
\STATE let $z'$ be an optimal extreme point to \ref{lp:choose}
\STATE $\cB^{(1)} \gets \{B'_{(i,r)} : z'_{(i,r)} > 0\} \cup \bigcup_{z'_{(i,r)} = 0} \{(i', 3 \cdot r') : (i',r') \in G_{(i,r)}\}$
\STATE $\cB^{(2)} \gets \{(i, 3 \cdot r) : (i,r) \in \cB_2\}$
\STATE let $\cB$ be $\{(i,3 \cdot r) : (i,r) \in \cB'\}$ plus the cheaper of the two sets $\cB^{(1)}$ and $\cB^{(2)}$
\STATE $\mathcal S \gets \mathcal S \cup \{\cB\}$
\ENDFOR
\RETURN the cheapest solution from $\mathcal S$
\end{algorithmic}
\end{algorithm}


\section{Minimum Sum of Diameters}\label{sec:msd}

It is easy to see that an $\alpha$-approximation for MSR yields an $\alpha$-approximation for MSD. That is, if we pick an arbitrary point to act as a center in each of the optimum MSD solution clusters, we see the optimum MSR solution has cost at most $OPT_{MSD}$.
So the $\alpha$-approximation will find an MSR solution with cost at most $\alpha \cdot OPT_{MSD}$. The diameter of any ball is at most twice its radius, so this yields a MSD solution of cost at most $2 \cdot \alpha \cdot OPT_{MSD}$.

In this section, we observe that a slight modification to the MSR approximation in fact yields a 6.546-approximation for MSD. Letting $\alpha = 3.389$ denote our approximation guarantee for MSR, note that $6.546 < 2 \cdot \alpha = 6.778$. That is, we are doing better than the trivial reduction of MSD to MSR that loses a factor of 2.


First note that for any $Y \subseteq X$ with diameter, say, $diam(Y)$, for any $i \in Y$ we have $Y \subseteq B(i, diam(Y))$ and $diam(B(i,diam(Y)) \leq 2 \cdot diam(Y)$.
So while it is difficult to guess any single cluster from the optimum MSD solution, we can guess the $1/\epsilon$ largest diameters (the values) and guess balls $\cB'$ with these radii that cover these largest-diameter clusters. Let $OPT'_D$ denote the total diameter of the remaining clusters from the optimum solution, $k' = k-\frac{1}{\epsilon}$, $X'$ be the remaining points to cluster, and $R_m = \min\{r : (i,r) \in \cB'\} \leq \epsilon \cdot OPT_D$.

For any $\lambda \geq 0$, note $OPT_{\text{\ref{lp:primal}}} + \lambda \cdot k' \leq OPT'_D$ as picking any single center from each cluster in optimum solution on $X'$ yields an MSR solution with cost at most $OPT'_D$.
We then use Theorem \ref{thm:bipoint} to get a bi-point solution $\cB_1, \cB_2, \lambda$.

If we triple the balls in $\cB_2$ and output those clusters, we get a solution with total diameter $\leq 6 \cdot cost(\cB_2)$.
For the other case, we again form groups $\mathcal G$.
Instead of picking a ball $B'_{(i,r)}$ for each group $G_{(i,r)} \in \mathcal G$, we simply let $B'_{(i,r)}$ be the set of points covered by the tripled balls in $G_{(i,r)}$. 

\begin{claim}
$diam(B'_{(i,r)}) \leq 2 \cdot r + 6 \cdot C_{(i,r)}$ where $B'_{(i,r)}$ is all points covered by the tripled balls from $G_{(i,r)}$.
\end{claim}
\begin{proof}
Consider any two points $j', j''$ covered by $\cup_{(i',r') \in G_{(i,r)}} B(i', 3\cdot r')$, say $(i',r')$ and $(i'',r'')$ are the balls in $G_{(i,r)}$ which, when tripled, cover $j'$ and $j''$, respectively. If $(i',r') = (i'',r'')$ (i.e. it is the same tripled ball from $G_{(i,r)}$ that covers both $j', j''$) then $d(j',j'') \leq 6 \cdot r' \leq 6 \cdot C_{(i,r)}$. Otherwise, we have $r' + r'' \leq C_{(i,r)}$ and
\[ d(j',j'') \leq d(j', i') + d(i',i) + d(i,i'') + d(i'',j'') \leq 4\cdot r' + r + r + 4 \cdot r'' \leq 2\cdot r + 4 \cdot C_{(i,r)}. \]
In either case, we can upper bound $d(j',j'') \leq 2 \cdot r + 6 \cdot C_{(i,r)}$, so $diam(B'_{(i,r)})$ is bounded by the same.
\end{proof}

We use an LP similar to \ref{lp:choose} except with the modified objective function to reflect the diameter costs of the corresponding choices.
\[
\begin{array}{rrcll}
{\bf minimize}: & \sum_{G_{(i,r)} \in \cG} (1-z_{(i,r)}) \cdot 6 \cdot C_{(i,r)} + z_{(i,r)} \cdot diam(B'_{(i,r)}) \tag{\bf{LP-Choose MSD}} \label{lp:choosemsd} \\
{\bf subject~to}: & \sum_{G_{(i,r)} \in \cG} \left((1-z_{(i,r)}) \cdot |G_{(i,r)}| + z_{(i,r)}\right)  \leq  k' \\
& z_{(i,r)}  \in  [0,1] \qquad \forall~ G_{(i,r)} \in \cG
\end{array}
\]

For $a,b\geq 0$, we let $a+b =1$ and $a\cdot k_{1} + b\cdot k_{2} = k'$, similar to MSR. Setting $z_{(i,r)} = b$ shows the optimum LP solution value is at most
\begin{eqnarray*}
& & \sum_{G_{(i,r)} \in \cG} a \cdot 6 \cdot C_{(i,r)} + b \cdot diam(B'_{(i,r)}) \\
& = & \sum_{G_{(i,r)} \in \cG} a \cdot 6 \cdot C_{(i,r)} + b \cdot 2 \cdot r + b \cdot 6 \cdot C_{(i,r)} \\
& = & (6a + 6b) \cdot \sum_{G_{(i,r)} \in \cG} C_{(i,r)} + 2b \cdot \sum_{G_{(i,r)} \in \cG}  \\
& = & 6 \cdot C_2 + 2 \cdot b \cdot C_1.
\end{eqnarray*}

In an optimal extreme point, at most one variable in \ref{lp:choosemsd} that is fractional so we set it to $1$ we pick to corresponding group to be covered by a single ball just like we did with the MSR algorithm.
 
Using the better of this solution or just tripling the balls in $\cB_2$ yields an MSD solution with total diameter at most
$\min\left\{6 \cdot C_2, 6 \cdot C_1 + 2 \cdot b \cdot C_2 + O(\epsilon) \cdot OPT_{D}\right\} \leq (1-d) \cdot 6 \cdot C_2 + d \cdot \left(b \cdot 2 \cdot C_2 + 6 \cdot C_1\right)$
 for any $d \in [0,1]$. Let 
 By setting $d =  \frac{6(1-b)}{2 \cdot b^{2} - 2 \cdot b + 6}$, the worst case analysis for the final bound happens when $b = 1/2$, at which we see the cost is at most $\frac{72}{11} \cdot OPT'_D + O(\epsilon) \cdot OPT_D$.
Adding this to the $1/\epsilon$ balls we guessed (whose diameters are at most twice their radius) and choosing $\epsilon$ sufficiently small shows
we get a solution with an approximation guarantee of 6.546 for MSD, which is better than two times the MSR guarantee. 

\section{Minimum Sum of Squared Radii}
The algorithm for MSSR follows the exact same procedure as MSR, the only modification happens at the phase where we analyze the cost of a group based on the center picked for it. 
The case by case analysis is similar to the MSR Case. 
However, note that the objective function value in the lagrangified LP is in the following form:
\begin{equation*}
    \sum_{(i,r)}x_{(i,r)}(r^{2} + \lambda)
\end{equation*}

Remember the upper-bounds for a single ball that is centerd in eitehr of $i_{1}, i_{2}$, and $i_{3}$, as depicted in \ref{fig:pic2}.
With these bounds, we now describe how to choose a single ball covering the points covered by tripled balls in $G_{(i,r)}$ in a way that gives a good bound on the minimum-radius ball covering these points. 

The optimization steps have been deferred to the Appendix B. 

We then find ourselves in two cases, 
\begin{enumerate}
    \item $R_{2} \geq \frac{1}{2}\cdot R_{1}$, then the cost is $9 \cdot R_{2} \leq 9 \cdot C(G_{(i_{1}, R_{1})})$, and
    \item $R_{2} \leq \frac{1}{2}\cdot R_{1}$, the cost is upper-bounded by $(\frac{27}{4}\cdot R_{1}^{2} + 9 \cdot R_{2}^{2})$.
\end{enumerate}
In any case, we see that by selecting $B'_{(i,r)}$ optimally, the radius is at most $\frac{27}{4}\cdot R_{1}^{2} + 9 \cdot R_{2}^{2}$. 
This is the solution where we aim to pick a group and merge it. 
Similar to before, we will compare the outcome of the corresponding \textit{Choose LP} with $C_{2}$, the solution in which we just tripled the radii of the output balls of Rounding.

Now we must again choose which groups to merge similar to both MSR and MSD case. 
Hence, consider the following LP-choose. 
\[
\begin{array}{rrcll}
{\bf minimize}: & \sum_{G_{(i,r)} \in \cG} (1-z_{(i,r)}) \cdot 3 \cdot C_{(i,r)} + z_{(i,r)} \cdot cost(\{B'_{(i,r)}\}) \tag{\bf{LP-Choose ({MMSR})}} \label{lp:choosemssr} \\
{\bf subject~to}: & \sum_{G_{(i,r)} \in \cG} \left((1-z_{(i,r)}) \cdot |G_{(i,r)}| + z_{(i,r)}\right)  \leq  k'& \\
 &z_{(i,r)}  \in  [0,1] & \forall~ G_{(i,r)} \in \cG
\end{array}
\]

To consolidate the groups, compute an optimal extreme point to \ref{lp:choosemssr}. 
Since all but one constraint are $[0,1]$ box constraints, there is at most one variable $z_{(i,r)}$ that does not take an integer value due to Rank Lemma. 
Since $|G_{(i,r)}| \geq 1$, then setting $z_{(i,r)}$ to 1 yields a feasible solution whose cost increases by at most the radius of $B'_{(i,r)}$, which was observed to be at most $(14 \cdot R_m)^{2} \leq 14^{2} \cdot \epsilon \cdot OPT$.

The final cost is $\min\left\{9 \cdot C_2, 9 \cdot C_1 + \frac{27}{4} \cdot b \cdot C_2 + O(\epsilon) \cdot OPT_{MSSR}\right\} \leq (1-d) \cdot 9 \cdot C_2 + d \cdot \left(b \cdot \frac{27}{4} \cdot C_2 + 9 \cdot C_1\right)$ for any $d \in [0,1]$. 

Then, we need to find $d$ such that we can upper-bound  $\min\left\{9 \cdot C_2, 9 \cdot C_1 + \frac{27}{4} \cdot b \right\}$ with an expression $\beta \cdot (aC_{1} + bC_{2})$. 
Note that the parametric equation for $\beta$
in the maximized point is
\begin{equation*}
    \beta = \frac{9d}{1-b}
\end{equation*}
which we set equal to 
\begin{equation*}
    \frac{9(1-d) + \frac{27}{4}bd}{b}.
\end{equation*}

Then, we can write $d = \frac{9(1-b)}{\frac{27}{4}b^{2} - \frac{27}{4}b + 9}$. 
Then, $\beta$ is maximized at $b = \frac{1}{2}$ and then $d$ is equal to $\frac{32}{55}$.

Then we proceed to construct the choose LP and at last, we get a 11.078-approximation algorithm.

%


\section{A Simple LMP Algorithm via Direct LP Rounding}\label{sec:lmp}
We again emphasize that one can slightly adapt the primal/dual algorithm and analysis in \cite{charika04} to prove a slightly weaker version of Theorem \ref{thm:bipoint} that would still suffice for our approximation guarantees. The main difference is that the averaging of the bipoint solution costs as given in bound \eqref{eqn:bound} from Section \ref{sec:combining} would be bounded by $(1+\epsilon') \cdot OPT$ for some $\epsilon' > 0$ (the running time depends linearly on $\log 1/\epsilon$). But we believe there is merit to our approach as it shows how may be able to cleanly avoid losing an $\epsilon$ in binary searching a Lagrangian relaxation when an LMP-algorithm is obtained through direct LP rounding.

The proof of Theorem \ref{thm:bipoint} proceeds through the usual approach of using a binary search using an LMP algorithm. We begin by describing our LMP algorithm followed by a simple ``upper-bound'' step which is used in some parts of the binary search.

Algorithm \ref{alg:round} describes the rounding procedure mentioned in Theorem \ref{thm:round}. Note it only depends on $x'$, the solution to $LP(\lambda)$ and not on $\lambda$ itself.
\begin{algorithm}
\caption{ROUND($x'$)}
\begin{algorithmic} \label{alg:round}
\STATE $\cB \gets \emptyset$
\FOR{$(i,r)$ with $x'_{(i,r)} > 0$ in non-increasing order of $r$}
\IF{$B(i,r) \cap B(i',r') = \emptyset$ for each $(i',r') \in \cB$}
\STATE $\cB \gets \cB \cup \{(i,r)\}$
\ENDIF
\ENDFOR
\RETURN $\cB$
\end{algorithmic}
\end{algorithm}


Let $x^{'}$ be the optimal solution for \ref{lp:primal}. It is important to remember that \eqref{lp:primal} and \eqref{lp:dual} only have variables/constraints for balls $B(i,r)$ with $i \in X'$ and $0 \leq r \leq 2 \cdot R_m$.
\begin{proof}[Proof of Theorem \ref{thm:round}]
Disjointedness follows by construction and clearly each ball in $\cB$ has $x_{i,r} > 0$. Consider some $j$ not covered by any ball $B(i,r)$ in $\cB$. Let $B(i',r')$ be any ball covering $j$ in the support of the LP, such a ball exists because every point is fractionally covered to an extent of at least $1$ in the LP. Since $B(i',r') \notin \cB$ there is some $B(i,r) \in \cB$ with $B(i',r') \cap B(i,r) \neq \emptyset$ and $r \geq r'$ (by how we ordered the balls in the algorithm). Let $j'$ be any point in $B(i',r') \cap B(i,r)$. Then $d(j,i) \leq d(j,i') + d(i',j') + d(j',i) \leq r' + r' + r \leq 3 \cdot r$. So $j \in B(i,3r)$, meaning if we tripled the radii of all balls in $\cB$ then $j$ would be covered.
\end{proof}

As noted in Corollary \ref{cor:round}, we have $cost(\cB) + \lambda \cdot |\cB| \leq \sum_{j \in X'} y'_j = OPT_{{\bf LP}(\lambda)}$.
Thus, we call this a ``Langrangian multipler preserving'' algorithm because if $\cB''$ is obtained by tripling the radii of the balls returned by ROUND$(x')$, then $cost(\cB'') + 3 \cdot \lambda \cdot |\cB''| \leq 3 \cdot OPT_{{\bf LP}(\lambda)}$.

\subsection{Binary Search Preparation: Starting Values and Filling a Solution}
To begin our binary search, we first. 
show that extreme values of $\lambda$ will yield solutions with $\geq k'$ and $\leq k'$ balls. This will let us set the initial window in the binary search.
\begin{lemma}\label{lem:initial}
Consider $\lambda \geq 0$ and an optimal solution $x'$ for \ref{lp:primal}. If $\lambda = 0$ then calling ROUND($x'$) will produce a solution with $|X'| > k'$ balls. If $\lambda = 2 \cdot k' \cdot R_m + 1$ then calling ROUND($x'$) will produce a solution with at most $k'$ balls.
\end{lemma}
\begin{proof}
First consider the case $\lambda = 0$. The LP solution that sets $x_{(i,0)} = 1$ for each $i \in X'$ and all other variables to 0 has value 0. It is also the only optimal solution as supporting any variable corresponding to a ball with positive radius yields an LP solution that has strictly positive cost. So $x'$ only supports balls with radius 0. Since $i' \notin B(i,0)$ for distinct $i,i' \in X'$ (as we assumed $d(i,i') > 0$), then ROUND($x'$) will return $|X'| > k'$ balls, one per point.

For brevity, let $\delta := 2 \cdot k' \cdot R_m$.
Now, we will show that if $\lambda = \delta + 1$ the rounding algorithm for \ref{lp:primal} returns at most $k'$ balls. Suppose otherwise, i.e. that the returned set of balls $\cB'$ has size exceeding $k'$ when $\lambda = \delta+1$.
Notice
\begin{equation*}
    (k'+1) \cdot \lambda \leq |\cB'| \cdot \lambda \leq  \sum_{(i,r) \in \cB'} r + \lambda\cdot|\cB'| \leq OPT_{\text{\ref{lp:primal}}} \leq \delta + k' \cdot \lambda
\end{equation*}
where the second last bound is from Corollary \ref{cor:round} and the last bound because setting $x_{(i,r)} = 1$ for every ball in the $k$-center solution from Theorem \ref{thm:summary} (and all other variables to $0$) is a feasible LP solution with cost at most $2 \cdot k' \cdot R_m + k' \cdot \lambda$. But this is a contradiction, since $\delta = \lambda - 1$.

\end{proof}

Finally, we require the following routine FILL($\cB_1, \cB_2$) which will be used to ensure the property from Theorem \ref{thm:bipoint} that all balls in the first solution intersect at least one ball from the second solution.
\begin{algorithm}
\caption{FILL($\cB_{1}, \cB_{2}$)}
\begin{algorithmic} \label{alg:fill}
\WHILE{Some ball $(i,r) \in \cB_1$ is disjoint from balls in $\cB_2$ and $|\cB_2| < k'$}
\STATE $\cB_2 \gets \cB_2 \cup \{(i,r)\}$
\ENDWHILE

\IF {$|\cB_2| = k'$}
\STATE $\cB_1 \gets \cB_2$
\ENDIF 
\RETURN $\cB_1, \cB_2$
\end{algorithmic}
\end{algorithm}

\begin{lemma}\label{lem:move}
Let $\cB_1, \cB_2$ be two sets obtained by rounding two optimal solutions $x_1, x_2$ to \ref{lp:primal} for some common value $\lambda$. For each of the final sets $\cB$ returned by FILL$(\cB_1, \cB_2)$, we have $cost(\cB) + \lambda \cdot |\cB| \leq OPT_{{\bf LP}(\lambda)}$. Furthermore, if $\cB'_1, \cB'_2$ denotes the pair of returned sets by FILL$(\cB_1, \cB_2)$, then we have $|\cB'_1| \geq k' \geq |\cB'_2|$. Finally, each $(i,r) \in \cB'_1$ intersects at least one ball in $\cB'_2$.
\end{lemma}
\begin{proof}
Let $y'$ be an optimal dual solution for this $\lambda$. Notice that both $x_1$ and $x_2$ will satisfy complementary slackness conditions along with $y'$ since $x_1$ and $x_2$ are both optimal primal solutions for the same LP and $y'$ is an optimal dual for this LP.

For a ball $(i,r)$ that is added to 
$\cB_2$ in the algorithm, since $(i,r)$ is disjoint from all balls in $\cB_2$ and since it is in the support of $x_1$ we maintain $r + \lambda = \sum_{j \in B(i,r)} y'_j$ and the invariant that balls in $\cB_2$ have their $r+\lambda$ value paid for by the dual values of disjoint subsets of points.

That $|\cB'_1| \geq k' \geq |\cB'_2|$ is immediate, given that the original sets $\cB_1, \cB_2$ satisfy this bound as well. The final condition that each ball in $\cB'_1$ intersect at least one ball in $\cB'_2$ is from the fact that either $\cB'_1 = \cB'_2$ (if $|\cB'_2| = k'$) or that the procedure stopped before $|\cB'_2|$ became $k'$ (i.e. there are no more balls in the original set $\cB_1$ that are disjoint from all balls in $\cB'_2$).
\end{proof}


\subsection{The Binary Search}
In this section, we complete the proof of Theorem \ref{thm:bipoint}.
Recall feasibility of some $x'$ for \ref{lp:primal} is not dependent on $\lambda$ itself since the constraints are independent of $\lambda$.
We call a value $\lambda > 0$ {\bf smooth} if for some $\epsilon > 0$ we have that the set of optimal extreme points to \ref{lp:primal} is the same as the set of optimal extreme points for ${\bf LP}(\lambda')$ for any $\lambda' \in [\lambda - \epsilon, \lambda + \epsilon]$. Otherwise,
we call $\lambda$ a {\bf break point}. We also call $\lambda = 0$ a break point.

We will prove that there is sufficiently large distance between consecutive break points. Our binary search algorithm will proceed until the window is small enough to enclose at most one break point, unless an earlier stopping criteria is met. At this point we can compute the break point itself and then return the required bi-point solution.

{\bf Invariant}: The binary search will maintain values $0 \leq \lambda_1 < \lambda_2$. For each $\ell = 1, 2$, let $x_\ell$ be an optimal solution to {\bf LP}($\lambda_\ell$) and let $\cB_\ell =\text{ROUND}(x_\ell)$. We also maintain $|\cB_1| \geq k' \geq |\cB_2|$ and that $x_1$ is {\em not} an optimal solution for {\bf LP}($\lambda_2$).

At each step in the binary search, we run the following check to ensure the last invariant holds.
\begin{itemize}
\item {\bf Check }: For a pair $\lambda_1, \lambda_2$, let $x_1, x_2$ be corresponding optimal LP solutions. If $x_1$ is optimal for {\bf LP}($\lambda_2$) we perform procedure in FILL on $\text{ROUND}(x_1)$ and $\text{ROUND}(x_2)$ (Lemma \ref{lem:move}) and return the resulting sets along with $\lambda_2$ as the sets from Theorem \ref{thm:bipoint}.
By properties of the balls from Lemma  \ref{lem:move} (noting $x_1$ is optimal for {\bf LP}($\lambda_2$)), the returned quantities satisfy the properties stated in Theorem \ref{thm:bipoint}.
\end{itemize}

Let $\Delta = 8 \cdot n \cdot n^{4 \cdot n^2}$, we will show $1/\Delta$ is a lower bound on the gap between break points (cf. Lemma \ref{lem:gap} below). Once $\lambda_2 - \lambda_1 \leq 1/\Delta$, then if the above check fails we have that the largest $\lambda$ for which
$x_1$ is an optimal solution to \ref{lp:primal} satisfies $\lambda \in [\lambda_1, \lambda_2]$. We show how to compute this value $\lambda$ exactly in Lemma \ref{lem:largest} and describe how to use this to find sets satisfying the requirements of Theorem \ref{thm:bipoint}.
The full binary search algorithm is described in Algorithm \ref{alg:bipoint}.

\begin{algorithm}
\caption{Binary Search to Find the B-ipoint Solution}
\begin{algorithmic} \label{alg:bipoint}
\STATE $\lambda_1 \gets 0, \lambda_2 \gets 2 \cdot R_m \cdot k' + 1$ and corresponding $x_1, \cB_1, x_2, \cB_2$ as in the invariant
\IF{{\bf Check} on $(\lambda_1, \lambda_2)$ passes}
\RETURN the corresponding solution
\ENDIF
\WHILE{$\lambda_1 + 1/\Delta < \lambda_2$}
\STATE Let $\lambda \gets (\lambda_1 + \lambda_2)/2$, $x'$ an optimal solution for \ref{lp:primal}, and $\cB$ be the output from ROUND($x'$).

\IF{$|\cB| \geq k'$}
\STATE $\lambda_1, x_1, \cB_1 \gets \lambda, x', \cB$ 
\ELSE
\STATE $\lambda_2, x_2, \cB_2 \gets \lambda, x', \cB$

\ENDIF
\IF{{\bf Check} on $(\lambda_1, \lambda_2)$ passes}
\RETURN the corresponding solution
\ENDIF
\ENDWHILE
\STATE Compute the only breakpoint $\lambda \in [\lambda_1, \lambda_2]$ (cf. Lemma \ref{lem:largest}), and let $\cB$ be the output for ROUND($x'$)

\IF{$|\cB| \geq k'$}
\STATE perform FILL($\cB, \cB_2$) and return the resulting sets along with $\lambda$
\ELSE
\STATE perform FILL($\cB_1, \cB$) and return the resulting sets along with $\lambda$
\ENDIF
\end{algorithmic}
\end{algorithm}

After the initial values $\lambda_1 = 0$ and $\lambda_2 = \delta + 1$ are set, if the checks all fail then the invariant is initially true. In each step of the search, if the checks fail then it is easy to see the invariant continues to hold.

If the loop terminates without returning, the final
 returned sets and $\lambda$-value have the properties stated in Theorem \ref{thm:bipoint}. This immediately follows from Lemma \ref{lem:move}, and the fact that $x_1$ and $x_2$ are both optimal for \ref{lp:primal} as $\lambda$ is the only break point in $[\lambda_1, \lambda_2]$.
So in $O(\log ((1+2 \cdot R_m \cdot k') \cdot \Delta)) = O(\log R_m + n^2 \log n)$ iterations, which is polynomial in the input size, the binary search will return a bi-point solution satisfying the properties stated in Theorem \ref{thm:bipoint}.

\subsubsection{Supporting Results for the Binary Search}

\begin{lemma}\label{lem:largest}
Let $x'$ be an optimal solution for {\bf LP}($\lambda_1$). In polynomial time, we can compute the greatest $\lambda$ such that $x'$ remains optimal for \ref{lp:primal}.
\end{lemma}
\begin{proof}
Consider the following LP for this fixed value of $x'$ but having $\lambda$ as a variable and variables $y_j, j \in X'$ as in \ref{lp:dual}.
\[
\begin{array}{rrcll}
{\bf maximize}: & \lambda & \\
{\bf subject~to}: & \sum_{j \in B(i,r) \cap X'} y_j & \geq & r + \lambda & \forall~ (i,r) \text{ s.t. } r \leq 2 \cdot R_m \\
& \sum_{j \in X'} y_j & = & \sum_{(i,r) : r \leq 2 \cdot R_m} x' \cdot (r + \lambda) \\
& y, \lambda & \geq & 0
\end{array}
\]
We emphasize that $x'$ is a fixed value in this setting, so the second constraint is linear in the variables $y$ and $\lambda$.

The first and third constraints assert $y$ is a feasible dual solution for the particular $\lambda$. The second asserts its value in \ref{lp:dual} is equal to the the value of $x_1$ in \ref{lp:primal}. Thus, $x'$ is optimal for \ref{lp:primal} exactly if there is some corresponding $y$ that causes all of the constraints of the above LP to hold. So solving this LP will yield the maximum $\lambda$ such that $x'$ is an optimal solution for \ref{lp:primal}.
\end{proof}

Recall Hadamard's bound on the determinant of a matrix in terms of the lengths of its row vectors. Given an $n \times n$ matrix $N$ where we let $N_i$ denote the $i$'th row of $N$, we have
\begin{equation*}
    |\det(N)| \leq \prod_{i=1}^{n}||N_{i}||_{2}
\end{equation*}

\begin{lemma}\label{lem:gap}
For two different break points $\lambda < \lambda'$, we have $\lambda + 1/\Delta < \lambda'$ where $\Delta = 8 \cdot n^2 \cdot n^{4 \cdot n^2}$.
\end{lemma}
\begin{proof}
Let $x$ be any extreme point solution of the polytope defining \ref{lp:primal}. So $x$ is the unique solution to a $M \cdot x = b$ where $M$ is an $n \times n$ non-singular submatrix of the constraint matrix (here, $n = |X'|$). From Cramer's rule, the denominator of each variable is bounded by $|\det(M)|$. Since the constraint matrix only has entries 0 and 1, each row $M_j$ satisfies $||M_j||_2 \leq \sqrt{n}$. By Hadamard's determinant bound, $|\det(M)| \leq \prod_j ||M_j||_2 \leq n^{n/2}$. Thus, every denominator in $x$ is an integer at most $n^{n/2}$.

Note $\lambda' > 0$. Let $\lambda''$ be very close to $\lambda'$ such that some extreme point $x'$ that is optimal for {\bf LP}($\lambda'$) is not optimal for {\bf LP}($\lambda''$).
This must be the case, it could not be that there is an extreme point that is optimal for $\lambda''$ arbitrarily close to $\lambda'$ but not for $\lambda'$ itself since the set of $\lambda''$ for which a particular $x$ is an optimal solution is a closed set. Let $x''$ be an optimal extreme point for {\bf LP}($\lambda''$), which then must be optimal for {\bf LP}($\lambda'$) as well.

Define a linear function $f'(z) = \sum_{(i,r)} (z + r) \cdot x'_{(i,r)}$ and similarly define $f''(z) = \sum_{(i,r)} (z + r) \cdot x''_{(i,r)}$. Then $f'(\lambda') = f''(\lambda')$ but $f'(\lambda'') \neq f''(\lambda'')$ so they have different slopes. That is, $f'(z) = f''(z)$ has a unique solution, namely at $z = \lambda' = \frac{\sum_{(i,r)} r \cdot (x''_{(i,r)} - x'_{(i,r)})}{\sum_{(i,r)} x'_{(i,r)} - x''_{(i,r)}}$. Note each of $x'$ and $x''$ supports at most $n$ values since they are extreme points of a polytope with only $n$ constraints apart from nonnegativity. So the top term in the ratio above expressing $\lambda'$ is a fraction of the form $N/D$ where $D \leq n^{n/2 \cdot 2n} = n^{n^2}$. Similarly, the bottom term of the ratio for $\lambda'$ is a fraction of the form $N'/D'$ where $N' \leq 2n \cdot n^{n^2}$ (using the fact that all $x'$ and $x''$ values are $\leq 1$). Thus, $\lambda'$ itself is a fraction whose denominator is at most $2n \cdot n^{2 \cdot n^2}$.

Finally, since $\lambda$ and $\lambda'$ are different break points, then $\lambda' - \lambda$ is a fraction whose denominator is at most $4n^2 \cdot n^{4 \cdot n^2} = \Delta/2$. Thus, $\lambda + 1/\Delta < \lambda'$.
\end{proof}

\bibliographystyle{unsrt}  
\bibliography{references}  

\appendix

\section*{Appendix}
\subsection*{Appendix A: Optimizing our choice of parameter for the MSR analysis}
\label{app:1}
The MSR analysis included a variety of cases and certain constants were chosen to define these cases. 
Here, we show that our choices of constants are optimal for our analysis techniques.

Let $\delta$ be the ratio of $R_{3}$ to $R_{2}$.
We wish to transform radius of the cluster picked at each center to be similar to other cases. 
In other words, we wish to transform each radius into a term like $(1 + \beta) \cdot R_{1} + (4 - \alpha) \cdot R_{2}$, a term similar to the radius of case of picking $i_{1}$, but with less weight on $R_{2}$ and more on $R_{1}$.
After finding the values to $\alpha$ and $\beta$ in terms of $\delta$, we make a decision for the center to pick based on value of $\frac{R_{3}}{R_{2}}$:
\begin{enumerate}
    \item soln*(1): Picking the center based on the value of $R_{2}$ according to table $\ref{table:2}$ if $\frac{R_{3}}{R_{2}} \leq \delta$.
    \item soln*(2): Picking the center at $i_{1}$ if $\frac{R_{3}}{R_{2}} \geq \delta$.
\end{enumerate}
Based on value of $\alpha$ and $\beta$, we can then find the parametric approximation factor for soln*(1) parameterized by $\delta$, which we do in what follows. 
But first, let us calculate the approximation factor for soln*(2) using the following inequality:
\begin{equation*}
    R_{2} + R_{3} \leq \textit{Cost($G_{i_{1},R_{1}}$)} \Longrightarrow (1 + \delta) \cdot R_{2} \leq \textit{Cost($G_{i_{1},R_{1}}$)} \Longrightarrow R_{2} \leq (\frac{1}{1 + \delta})\textit{Cost($G_{i_{1},R_{1}}$)}.
\end{equation*}
Then, the cost of a group is $R_{1} + (\frac{4}{1+\delta})\cdot \textit{Cost($G_{i_{1},R_{1}}$)} = \textit{Cost($B(i_{1},R_{1})$)} + (\frac{4}{1+\delta})\cdot \textit{Cost($G_{i_{1},R_{1}}$)}$ as $\textit{Cost($B(i_{1},R_{1})$)} = R_{1}$.
For the sake of brevity, let \textit{Cost($G_{i_{1},R_{1}}$)} = $C(G_{i_{1},R_{1}})$ and \textit{Cost($B(i_{1},R_{1})$)} = $C(B(i_{1},R_{1}))$.

Now, let us restate the cost of each scenario in terms of $\alpha$ and $\delta$, so that we can minimize the cost after aggregating the costs over all groups.  

\emph{Case 1 ($R_{2} \geq R_{1} + 2 \cdot \delta \cdot R_{2}$): }
\begin{itemize}
    \item Cost of picking $i_{1}$: $R_{1} + 4 \cdot R_{2}$.
    \item Cost of picking $i_{2}$: $3 \cdot R_{2}$ since $R_{2} \geq R_{1} + 2 \cdot \delta \cdot R_{2} \Longrightarrow 3R_{2} \geq 2R_{1} + (1 + 4 \cdot \delta)R_{2}$.
    \item Cost of picking $i_{3}$: $4R_{2}$ since $R_{2} \geq R_{1} + 2 \cdot \delta \cdot R_{2} \Longrightarrow 2R_{2} \geq 2R_{1} + 4 \cdot \delta R_{2} \Longrightarrow 4R_{2} \geq 2R_{1} + 4 \cdot \delta R_{2}$.  
\end{itemize}
Hence, we pick $i_{2}$ with minimum cost. 

\emph{Case 2 ($R_{1} + 2 \cdot \delta \cdot R_{2} \geq R_{2} \geq \frac{2}{3} \cdot R_{1} + \frac{4}{3} \cdot \delta \cdot R_{2}$): }
\begin{itemize}
    \item Cost of picking $i_{1}$: $R_{1} + 4 \cdot R_{2}$.
    \item Cost of picking $i_{2}$: $2 \cdot R_{1} + (1+4\cdot \delta)\cdot R_{2}$ since $R_{2} \leq R_{1} + 2 \cdot \delta \cdot R_{2} \Longrightarrow 3R_{2} \leq 2R_{1} + (1 + 4 \cdot \delta)R_{2}$.
    \item Cost of picking $i_{3}$: $4R_{2}$ since $R_{2} \geq \frac{2}{3} \cdot R_{1} + \frac{4}{3} \cdot \delta \cdot R_{2} \Longrightarrow 3R_{2} \geq 2  \cdot R_{1} + 4 \cdot \delta R_{2} \Longrightarrow 4R_{2} \geq 2R_{1} + 4 \cdot \delta R_{2}$.  
\end{itemize}
Note that $\frac{2}{3} \cdot R_{1} + \frac{4}{3} \cdot \delta \cdot R_{2} \leq R_{2} \Longrightarrow 2 \cdot R_{1} + (1+4\cdot \delta)\cdot R_{2} \leq 4 \cdot R_{2}$.
Also, we have $\frac{2}{3} \cdot R_{1} + \frac{4}{3} \cdot \delta \cdot R_{2} \leq R_{2} \Longrightarrow \frac{1}{3} \cdot R_{1} + \frac{4}{3} \cdot \delta \cdot R_{2} \leq R_{2} \longrightarrow R_{1} + 4 \cdot \delta \cdot R_{2} \leq 3 \cdot R_{1} \Longrightarrow  2 \cdot R_{1} + (1 +4 \cdot \delta) \cdot R_{2} \leq R_{1} + 4 \cdot R_{1}$.
Hence, we pick $i_{2}$ with minimum cost. 

Since $R_{2} \geq \frac{2}{3} \cdot R_{1} + \frac{4 \cdot \delta}{3} \cdot R_{2}$, then
\begin{equation*}
    \frac{-4 \cdot \delta + 3 - \alpha}{1 - \frac{4}{3} \cdot \delta}(R_{2} \cdot (1 - \frac{4}{3} \cdot \delta) - \frac{2}{3} \cdot R_{1}) \geq 0,
\end{equation*}
where $\delta < \frac{3}{4}$ and $3 \geq 4 \cdot \delta + \alpha$, for some value of $\alpha$.
Since we pick the center $i_{2}$, the cost is $2 \cdot R_{1} + (1+4 \cdot \delta) \cdot R_{2}$.
At last, 
\begin{equation*}
    \text{Cost } \leq 2 \cdot R_{1} + (1 + 4 \cdot \delta) R_{2} +     \frac{-4 \cdot \delta + 3 - \alpha}{1 - \frac{4}{3} \cdot \delta}(R_{2} \cdot (1 - \frac{4}{3} \cdot \delta) - \frac{2}{3} \cdot R_{1}) = (\frac{2 \cdot \alpha}{3 - 4 \cdot \delta}) R_{1} + (4 - \alpha) \cdot R_{2}.
\end{equation*}

\emph{Case 3 ($\frac{2}{3} \cdot R_{1} + \frac{4}{3} \cdot \delta \cdot R_{2} \geq R_{2} \geq \frac{1}{2} \cdot R_{1} + \delta \cdot R_{2}$): }
\begin{itemize}
    \item Cost of picking $i_{1}$: $R_{1} + 4 \cdot R_{2}$.
    \item Cost of picking $i_{2}$: $2 \cdot R_{1} + (1+4\cdot \delta)\cdot R_{2}$ since $R_{2} \leq R_{1} + 2 \cdot \delta \cdot R_{2} \Longrightarrow 3R_{2} \leq 2R_{1} + (1 + 4 \cdot \delta)R_{2}$.
    \item Cost of picking $i_{3}$: $4R_{2}$ since $R_{2} \geq \frac{2}{3} \cdot R_{1} + \frac{4}{3} \cdot \delta \cdot R_{2} \Longrightarrow 3R_{2} \geq 2  \cdot R_{1} + 4 \cdot \delta R_{2} \Longrightarrow 4R_{2} \geq 2R_{1} + 4 \cdot \delta R_{2}$.  
\end{itemize}
Note that $\frac{2}{3} \cdot R_{1} + \frac{4}{3} \cdot \delta \cdot R_{2} \geq R_{2} \Longrightarrow 2 \cdot R_{1} + (1+4\cdot \delta)\cdot R_{2} \geq 4 \cdot R_{2}$
Hence, we pick $i_{3}$ with minimum cost. 

Since $R_{2} \leq \frac{2}{3} \cdot R_{1} + \frac{4 \cdot \delta}{3} \cdot R_{2}$, then
\begin{equation*}
    \frac{\alpha}{1 - \frac{4}{3} \cdot \delta}(R_{2} \cdot (\frac{4}{3} \cdot \delta - 1) + \frac{2}{3} \cdot R_{1}) \geq 0,
\end{equation*}
where $\delta < \frac{3}{4}$ and $\alpha \geq 0$, for some value of $\alpha$.
Since we pick the center $i_{2}$, the cost is $4 \cdot R_{2}$.
At last, 
\begin{equation*}
    \text{Cost } \leq 4 \cdot R_{2} +     \frac{\alpha}{1 - \frac{4}{3} \cdot \delta}(R_{2} \cdot (\frac{4}{3} \cdot \delta - 1) + \frac{2}{3} \cdot R_{1}) = (\frac{2 \cdot \alpha}{3 - 4 \cdot \delta}) R_{1} + (4 - \alpha) \cdot R_{2}.
\end{equation*}

\emph{Case 4 ($\frac{1}{2} \cdot R_{1} + \delta \cdot R_{2} \geq R_{2} \geq \frac{1}{4} \cdot R_{1} + \delta \cdot R_{2}$): }
\begin{itemize}
    \item Cost of picking $i_{1}$: $R_{1} + 4 \cdot R_{2}$.
    \item Cost of picking $i_{2}$: $2 \cdot R_{1} + (1+4\cdot \delta)\cdot R_{2}$ since $R_{2} \leq R_{1} + 2 \cdot \delta \cdot R_{2} \Longrightarrow 3R_{2} \leq 2R_{1} + (1 + 4 \cdot \delta)R_{2}$.
    \item Cost of picking $i_{3}$: $2R_{1} + 4 \cdot \delta \cdot R_{2}$ since $R_{2} \leq \frac{1}{2} \cdot R_{1} + \delta \cdot R_{2} \Longrightarrow 4R_{2} \geq 2  \cdot R_{1} + 4 \cdot \delta R_{2}$.  
\end{itemize}
We have $\frac{1}{4} \cdot R_{1} + \delta \cdot R_{2} \leq R_{2} \Longrightarrow  R_{1} + 4 \cdot \delta \cdot R_{2} \leq 4 \cdot R_{2} \longrightarrow 2 \cdot R_{1} + 4 \cdot \delta \cdot R_{2} \leq R_{1} + 4 \cdot R_{2}$.
Hence, we pick $i_{3}$ with minimum cost. 

Since $R_{2} \geq \frac{1}{4} \cdot R_{1} + \delta \cdot R_{2}$, then
\begin{equation*}
    \frac{4-4\cdot \delta - \alpha}{1 - \delta}(R_{2} \cdot (1 - \delta) - \frac{1}{4} \cdot R_{1}) \geq 0,
\end{equation*}
where $4 \geq 4 \cdot \delta + \alpha$ and $\alpha \geq 0$, for some value of $\alpha$.
Since we pick the center $i_{3}$, the cost is $2 \cdot R_{1} + 4 \cdot \delta\cdot R_{2}$.
At last, 
\begin{equation*}
    \text{Cost } \leq 2 \cdot R_{1} + 4 \cdot \delta \cdot R_{2}    +  \frac{4-4\cdot \delta - \alpha}{1 - \delta}(R_{2} \cdot (1 - \delta) - \frac{1}{4} \cdot R_{1}) = (1+\frac{\alpha}{4 - 4 \cdot \delta}) R_{1} + (4 - \alpha) \cdot R_{2}.
\end{equation*}

\emph{Case 5 ($\frac{1}{4} \cdot R_{1} + \delta \cdot R_{2}$): }
\begin{itemize}
    \item Cost of picking $i_{1}$: $R_{1} + 4 \cdot R_{2}$.
    \item Cost of picking $i_{2}$: $2 \cdot R_{1} + (1+4\cdot \delta)\cdot R_{2}$ since $R_{2} \leq R_{1} + 2 \cdot \delta \cdot R_{2} \Longrightarrow 3R_{2} \leq 2R_{1} + (1 + 4 \cdot \delta)R_{2}$.
    \item Cost of picking $i_{3}$: $2R_{1} + 4 \cdot \delta \cdot R_{2}$ since $R_{2} \leq \frac{1}{2} \cdot R_{1} + \delta \cdot R_{2} \Longrightarrow 4R_{2} \geq 2  \cdot R_{1} + 4 \cdot \delta R_{2}$.  
\end{itemize}
We have $\frac{1}{4} \cdot R_{1} + \delta \cdot R_{2} \geq R_{2} \Longrightarrow  R_{1} + 4 \cdot \delta \cdot R_{2} \geq 4 \cdot R_{2} \longrightarrow 2 \cdot R_{1} + 4 \cdot \delta \cdot R_{2} \geq R_{1} + 4 \cdot R_{2}$.
Hence, we pick $i_{1}$ with minimum cost. 

Since $R_{2} \leq \frac{1}{4} \cdot R_{1} + \delta \cdot R_{2}$, then
\begin{equation*}
    \frac{\alpha}{1 - \delta}(R_{2} \cdot (\delta-1) + \frac{1}{4} \cdot R_{1}) \geq 0,
\end{equation*}
where $\delta < 1$ and $\alpha \geq 0$, for some value of $\alpha$.
Since we pick the center $i_{1}$, the cost is  $R_{1} + 4\cdot R_{2}$.
At last, 
\begin{equation*}
    \text{Cost } \leq  R_{1} + 4  \cdot R_{2}    +\frac{\alpha}{1 - \delta}(R_{2} \cdot (\delta-1) + \frac{1}{4} \cdot R_{1})  = (1+\frac{\alpha}{4 - 4 \cdot \delta}) R_{1} + (4 - \alpha) \cdot R_{2}.
\end{equation*}

\begin{table}[ht!]
\centering
\begin{tabular}{||c c c||} 
 \hline
 Case \# & $\frac{R_{3}}{R_{2}} < \delta$ & $\frac{R_{3}}{R_{2}} \geq  \delta$ \\ [0.5ex] 
 \hline\hline
 1 &  $3 \cdot C(G_{i_{1},R_{1}})$ & $\frac{3}{1+\delta} \cdot C(G_{i_{1},R_{1}})$ \\ 
 2 &  $\frac{2 \cdot \alpha}{3 - 4 \cdot \delta}\cdot C(B(i_{1},R_{1})) + (4 - \alpha) \cdot C(G_{i_{1},R_{1}})$ & $C(B(i_{1},R_{1})) + \frac{4}{1+\delta}\cdot C(G_{i_{1},R_{1}})$ \\
 3 &  $\frac{2 \cdot \alpha}{3 - 4 \cdot \delta}\cdot C(B(i_{1},R_{1})) + (4 - \alpha) \cdot C(G_{i_{1},R_{1}})$ & $C(B(i_{1},R_{1})) + \frac{4}{1+\delta}\cdot C(G_{i_{1},R_{1}})$  \\
 4 &  $(1+\frac{\alpha}{4 - 4 \cdot \delta})\cdot  C(B(i_{1},R_{1})) + (4 - \alpha) \cdot C(G_{i_{1},R_{1}})$ & $C(B(i_{1},R_{1})) + \frac{4}{1+\delta}\cdot C(G_{i_{1},R_{1}})$ \\
 5 &  $(1+\frac{\alpha}{4 - 4 \cdot \delta})\cdot C_{1}^{g} + (4 - \alpha) \cdot C(G_{i_{1},R_{1}})$ & $C(B(i_{1},R_{1})) + \frac{4}{1+\delta}\cdot C(G_{i_{1},R_{1}})$ \\ [1ex] 
 \hline
\end{tabular}
\caption{The cost of our choice for centers in every condition discussed.}
\label{table:2}
\end{table}
As per Table \ref{table:2}, and the conditions for $\alpha$ and $\delta$ in every case, set $\alpha = 1$ and $\delta = \frac{1}{3}$.

Then, the cost of each group is as in Table \ref{table:3} based on our decision:
\begin{table}[ht!]
\centering
\begin{tabular}{||c c c||} 
 \hline
 Case \# &  $\frac{R_{3}}{R_{2}} < \frac{1}{3}$ & $\frac{R_{3}}{R_{2}} \geq  \frac{1}{3}$ \\ [0.5ex] 
 \hline\hline
 1 &  $3 \cdot C(G_{i_{1},R_{1}})$ & $\frac{9}{4} \cdot C(G_{i_{1},R_{1}})$ \\ 
 2 &  $\frac{6}{5}\cdot C(B(i_{1},R_{1})) + (3) \cdot C(G_{i_{1},R_{1}})$ & $C(B(i_{1},R_{1})) + 3\cdot C(G_{i_{1},R_{1}})$ \\
 3 &  $\frac{6}{5}\cdot C(B(i_{1},R_{1})) + (3) \cdot C(G_{i_{1},R_{1}})$ & $C(B(i_{1},R_{1})) + 3\cdot C(G_{i_{1},R_{1}})$  \\
 4 &  $(\frac{11}{8})\cdot C(B(i_{1},R_{1})) + (3) \cdot C(G_{i_{1},R_{1}})$ & $C(B(i_{1},R_{1})) + 3\cdot C(G_{i_{1},R_{1}})$ \\
 5 &  $(\frac{11}{8})\cdot C(B(i_{1},R_{1})) + (3) \cdot C(G_{i_{1},R_{1}})$ & $C(B(i_{1},R_{1})) + 3\cdot C(G_{i_{1},R_{1}})$ \\ [1ex] 
 \hline
\end{tabular}
\caption{The cost of our choice for centers in every condition discussed.}
\label{table:3}
\end{table}
Note that each group's cost is in form of $A \cdot C(B(i_{1},R_{1})) + B \cdot C(G_{i_{1},R_{1}})$ where the maximum values for $A$ and $B$ happens when $\frac{R_{3}}{R_{2}} < \frac{1}{3}$ and $R_{2} \leq \frac{R_{1}}{2} + \frac{R_{2}}{3}$ with cost $\frac{11}{8} \cdot C(B(i_{1},R_{1})) +  3 \cdot C(G_{i_{1},R_{1}})$. 

\subsection*{Appendix B:Optimizing our choice of parameter for the MSSR analysis}
\label{app2}

We have summarized the case, the condition, and the center picked in Table .
\begin{table}[ht!]
\centering
\begin{tabular}{||c c c||} 
 \hline
 Case \# &  Condition & center picked \\ [0.5ex] 
 \hline\hline
 1 & $R_{2} \geq R_{1} + 2\delta R_{2}$ & $i_{2}$\\ 
 2 & $R_{1} + 2\delta R_{2} \geq R_{2} \geq \frac{2}{3}\cdot R_{1} + \frac{4}{3}\delta \cdot R_{2}$ & $i_{2}$\\ 
 3 & $\frac{2}{3}\cdot R_{1} + \frac{4}{3}\delta \cdot R_{2}\geq R_{2} \geq \frac{1}{2}\cdot R_{1} + \delta \cdot R_{2}$ & $i_{3}$\\ 
4 & $\frac{1}{2}\cdot R_{1} + \delta \cdot R_{2} \geq R_{2} \geq \frac{1}{4}\cdot R_{1} + \delta \cdot R_{2}$ & $i_{3}$\\ 
 5 & $\frac{1}{2}\cdot R_{1} + \delta \cdot R_{2} \geq R_{2} \geq \frac{1}{4}\cdot R_{1} + \delta \cdot R_{2}$ & $i_{1}$\\  [1ex] 
 \hline
\end{tabular}
\caption{The centers that will be picked in the Minimum Sum of Squared Radii based on values of $R_{1}$ and $R_{2}$ as depicted in \ref{fig:pic2}}
\label{table:4}
\end{table}
Then, Table describes the upper-bound on the cost of each choice, using the upper-bound value for $R_{2}$.
The main difficulty in this scenario is getting rid of $R_{1}R_{2}$ terms in the costs for each case.

\emph{Case 1, $R_{2} \geq R_{1} + 2\delta R_{2}$: }
Cost of picking $i_{2}$ will be $(3\cdot R_{2})^2 = 9\cdot R_{2}^{2}$.
We aim to balance the cost of other cases such that the upper-bound of each has one term $9\cdot R_{2}^{2}$ and a term with some coefficient for $R_{1}^{2}$.

\emph{Case 2, $R_{1} + 2\delta R_{2} \geq R_{2}$: }
We will be using the bound $R_{2} \leq 
\frac{1}{1-2\delta}R_{1}$.
Note that $\delta$ is the ratio $\frac{R_{3}}{R_{2}}$.
Then we must have that $\delta < \frac{1}{2}$. 
Then the cost is bounded by $4 + \frac{32\cdot \delta^{2} - 16 \cdot \delta - 5}{(1-2\delta)^{2}}R_{1}^{2} + 9R_{2}^{2}$.

\emph{Case 3, $\frac{2}{3}\cdot R_{1} + \frac{4}{3}\delta \cdot R_{2}$: }
We will be using the bound $R_{2} \leq 
\frac{2}{3-4\delta}R_{1}$.
Then we must have that $\delta < \frac{3}{4}$. 
At last,  the cost is bounded by $7 \cdot \frac{2}{3-4\cdot \delta}R_{1}^{2} + 9R_{2}^{2}$.

\emph{Case 4, $R_{1} + 2\delta R_{2} \geq R_{2}$: }
We will be using the bound $R_{2} \leq 
\frac{1}{2-2\delta}R_{1}$.
Then we must have that $\delta < 1$. 
Then the cost is bounded by $4 + \frac{-16\cdot \delta^{2} + 32 \delta - 9 }{(2-2\delta)^{2}}R_{1}^{2} + 9R_{2}^{2}$.

\emph{Case 5, $\frac{1}{2}\cdot R_{1} + \delta \cdot R_{2} \geq R_{2}$: }
We will be using the bound $R_{2} \leq 
\frac{1}{4-4\delta}R_{1}$.
Then we must have that $\delta < 1$. 
At last,  the cost is bounded by $ \frac{39 - 32\cdot\delta}{(4 - 4 \cdot \delta)^{2}}R_{1}^{2} + 9R_{2}^{2}$.

\begin{table}[ht!]
\centering
\begin{tabular}{||c c c||} 
 \hline
 Case \# &  Condition & center picked \\ [0.5ex] 
 \hline\hline
 1 & $R_{2} \geq R_{1} + 2\delta R_{2}$ & $i_{2}$\\ 
 2 & $R_{1} + 2\delta R_{2} \geq R_{2} \geq \frac{2}{3}\cdot R_{1} + \frac{4}{3}\delta \cdot R_{2}$ & $i_{2}$\\ 
 3 & $\frac{2}{3}\cdot R_{1} + \frac{4}{3}\delta \cdot R_{2}\geq R_{2} \geq \frac{1}{2}\cdot R_{1} + \delta \cdot R_{2}$ & $i_{3}$\\ 
4 & $\frac{1}{2}\cdot R_{1} + \delta \cdot R_{2} \geq R_{2} \geq \frac{1}{4}\cdot R_{1} + \delta \cdot R_{2}$ & $i_{3}$\\ 
 5 & $\frac{1}{2}\cdot R_{1} + \delta \cdot R_{2} \geq R_{2} \geq \frac{1}{4}\cdot R_{1} + \delta \cdot R_{2}$ & $i_{1}$\\  [1ex] 
 \hline
\end{tabular}
\caption{The centers that will be picked in the Minimum Sum of Squared Radii based on values of $R_{1}$ and $R_{2}$ as depicted in \ref{fig:pic2}}
\label{table:5}
\end{table}
\end{document}